\documentclass[11pt]{article}

 \usepackage{algorithm}
\usepackage{tcolorbox}
\usepackage{amssymb,xspace,graphicx,relsize,bm,xcolor,amsmath,breqn,algpseudocode,multirow}
\usepackage{tikz,amsmath, subcaption}
\usepackage{tcolorbox}
\newcommand{\Exp}{\mathbb{E}}

\usepackage[margin=1in]{geometry}

\newcommand{\M}{\mathcal{M}}
\newcommand{\suppress}[1]{}

\newcommand{\Cc}{\ensuremath{\mathcal{C}}}
\def\01{\{0,1\}}
\newcommand{\pmset}[1]{\{-1,1\}^{#1}} %

\usepackage{graphicx}
\usepackage{amsmath}
\usepackage{amssymb}
\usepackage{amsthm}
\usepackage{dsfont}
\usepackage{array}
\usepackage{makecell}

\newcommand{\A}{\ensuremath{\mathcal{A}}}

\newcommand{\R}{\ensuremath{\mathbb{R}}}

\newcommand{\id}{\ensuremath{\mathbb{I}}}

\newcommand{\Hi}{\ensuremath{\mathcal{H}}}
\usepackage{tikz,tikz-qtree}
\usepackage{wrapfig}
\usetikzlibrary{arrows}
\newcommand{\RSOA}{\ensuremath{\mathsf{RSOA}}}
\newcommand{\poly}{\ensuremath{\mathsf{poly}}}

\newcommand{\sfat}{\ensuremath{\mathsf{sfat}}}
\newcommand{\DP}{\ensuremath{\mathsf{DP}}}
\newcommand{\PAC}{\ensuremath{\mathsf{PAC}}}

\usepackage[pagebackref]{hyperref}
\usepackage[usestackEOL]{stackengine}
\usepackage{thm-restate,mathrsfs}
\usepackage{enumerate}
\usepackage{array}
\usepackage{parskip}
\def\01{\{0,1\}}
\newcommand{\ket}[1]{|#1\rangle}
\newcommand{\bra}[1]{\langle#1|}
\newcommand{\ketbra}[2]{|#1\rangle\langle#2|}
\hypersetup{
	colorlinks,
	linkcolor={blue!100!black},
	citecolor={red!100!black},
}

\newcommand{\be}{\begin{equation}}
\newcommand{\ee}{\end{equation}}
\newcommand{\ba}{\begin{array}}
\newcommand{\ea}{\end{array}}
\newcommand{\bea}{\begin{eqnarray}}
\newcommand{\eea}{\end{eqnarray}}

\newcommand{\braketIP}[2]{\langle{#1}|{#2}\rangle}

\usepackage{mathtools}

\DeclareMathOperator{\tr}{Tr}
\newcommand{\ra}{\rangle}
\newcommand{\la}{\langle}

\newcommand{\err}{\mathrm{err}}

\newcommand{\rank}{\mathrm{rank}}

\newcommand{\Tr}{\textsf{Tr}}

\newcommand{\relent}[2]{\mathrm{S}\left(#1\|#2\right)}

\newcommand{\calE}{{\cal E }}
\newcommand{\calC}{{\cal C }}

\newcommand{\FF}{\mathbb{F}}
\newcommand{\EE}{\mathbb{E}}
\newcommand{\ZZ}{\mathbb{Z}}

\newtheorem{question}{Question}

\newtheorem{theorem}{Theorem}

\global\long\def\argmin{\operatornamewithlimits{argmin}}

\newcommand{\MPS}[3]{#1^{(#2)}_{#3}}
\title{A survey on the complexity of learning quantum states}

\date{}

    \author{
    Anurag Anshu\\[2mm]
    Harvard
    University\\
    \small \texttt{anuraganshu@fas.harvard.edu} 
    \and
    Srinivasan Arunachalam\\[2mm]
     IBM Quantum, Almaden Research Center\\
    \small \texttt{Srinivasan.Arunachalam@ibm.com}
    }
\begin{document}
\maketitle
%%%%%%%%%%%%%%%%%%%%%%%%%%%%%%%%%%%%%%%%%%%%%%%%%%%%%%%%%%%%%%%
\begin{abstract}
 We survey various recent results that rigorously study the complexity of learning quantum states. 
 These include progress on quantum tomography, learning physical quantum states, alternate learning models to tomography and learning classical functions encoded as quantum states.  We highlight how these results are paving the way for a highly successful theory with a range of exciting open questions. To this end, we distill 25 open questions from these results. 
 %In fact, we give a gift award of \$ 50 each, funded by the affiliation of the second author. 
\end{abstract}

\section{Introduction}

   In the last decade, machine learning has received tremendous attention with the success of deep neural networks (or in more generality deep learning)  in practically relevant tasks such as natural language processing, speech recognition and image processing. Some popular applications of deep learning include AlphaGo  and AlphaZero (to play the games of Go and chess), chatGPT (to mimic a human conversation) and AlphaFold (for solving instances of protein folding)~\cite{jumper2021highly,vaswani2017attention,silver2017mastering}.  Although these machine learning techniques work very well in practice, they are not well understood from a theoretical standpoint. In a seminal work in 1984, Valiant~\cite{DBLP:journals/cacm/Valiant84} introduced the well-known probability approximately correct ($\PAC$) model of learning, which laid the mathematical foundation to understand machine learning from a computational complexity theory perspective. Since then, several mathematical models for machine learning have been proposed, some of which have theoretically justified the successes of practical learning algorithms. The study of machine learning from this complexity theoretic perspective is often referred to as \emph{computational learning~theory}. 
   
In another line of research, a century old quest which includes physicists, mathematicians and - now - computer scientists, is understanding the dividing line between simple and complex quantum states. Some prominent measures of complexity have been formulated in this process - for instance, correlation length and entanglement entropy \cite{EisertP10} from the physics point of view; quantum circuit size and description size~\cite{aaronson2016complexity} from the computer science point of view. A recent revolution in quantum information - inspired by practical implementations of quantum devices and incredible success of machine learning - has brought another measure in picture: \emph{learnability}. In the last decade, there have been several works to understand what classes of quantum states are learnable efficiently and why some classes of states are hard to learn. Here we argue that learnability as a complexity-theoretic metric is remarkably powerful and has been revealing fundamentally new properties of physically and computationally relevant quantum states. This is akin to the aforementioned $\PAC$ learning framework  used to understand machine learning from a  complexity theoretic framework.

A general formalism for learning quantum states is as follows. A learning algorithm (which we often refer to as a learner) receives many independent copies of an unknown quantum state - guaranteed to be within a ``class" of states (known to the learner). Using quantum measurements, the learner extracts information about the unknown state, and then outputs a sufficiently accurate description of the quantum state. We stress on the three defining notions in this general framework: the class of states, the type of measurement done by the learner and the metric for accuracy.
%$(i)$ what is the information the learner receives? $(ii)$  what is the class of states? $(iii)$ what does well-enough mean?
Modifying any one of these parameters can change the quantum learning model in an interesting way and we discuss these models in this survey. The complexity metric associated with these learning models is the quantum \emph{sample complexity}, defined as the number of copies of the unknown state used by the learning algorithm and quantum \emph{time complexity}, defined as the total number of gates used by the algorithm.\footnote{In this survey, we will also discuss classical sample and time complexity and it's definition will be clear when we discuss these complexities.}

%This is the general tomography problem. Here the goal  is to learn an arbitrary unknown quantum state $\rho$ up to  trace distance $\varepsilon$ given copies of $\rho$. The goal is to reduce the number of copies of $\rho$ and also the time Since $\rho$ is an arbitrary, this 

\subsection{Organization of this survey} 
Our survey discusses these learning models that come with rigorous guarantees on the sample and time complexity, as detailed below. 
\begin{enumerate}
    \item \emph{Learning arbitrary quantum states.} Here the goal is to learn an arbitrary  $n$-qubit quantum state $\rho$, given copies of $\rho$, up to  small trace distance. Given the generality of this task, the sample complexity  of this task is known to be exponential in $n$. We discuss this in Section~\ref{sec:tomography}.
    \item \emph{Learning physical quantum states.} A natural followup question is, can we learn \emph{interesting} subclasses of quantum states efficiently? In this direction we look at stabilizer states, states from the Clifford hierarchy, Gibbs states at different temperature regimes and matrix product states. We discuss this in Section~\ref{sec:physicalstates}. 
    
    \item \emph{Learning states in alternate models.} Suppose the goal of the learner was to still learn an unknown quantum state, can we weaken the requirement for the learner and still learn the unknown~$\rho$?  To this end, there are models of learning called $\PAC$ learning, online learning, shadow tomography and several equivalences between them. We discuss this in Section~\ref{sec:alternatemodels}.
    \item \emph{Learning classical functions encoded as states.} Suppose the unknown state $\rho$ encodes a classical function, what is the complexity of learning? Here we discuss known results on quantum $\PAC$ learning, agnostic learning, statistical query learning and kernel methods which encode classical data into quantum states, and exhibit the strengths and weaknesses of  quantum examples for learning classical functions.  We discuss this in Section~\ref{sec:classicalfunctionlearning}
\end{enumerate}
Finally we conclude in Section~\ref{sec:prospective} with some perspective on other works related to sample and time complexity of learning quantum states. Throughout this  survey we have put together several open questions that would improve our understanding on the complexity of quantum states from the perspective of learning theory.

%\paragraph{Organization} In Section~\ref{sec:tomography} we first discuss  tomography of arbitrary quantum states, in Section~\ref{sec:physicalstates} we discuss tomography of structured quantum states, in Section~\ref{sec:alternatemodels} , in Section~\ref{sec:classicalfunctionlearning} we discuss learning quantum states that encode a classical function and 

\section{Tomography}
\label{sec:tomography}
Quantum state tomography ($\textsf{QST}$) is the following task: given many independent copies of an unknown $n$-qubit quantum state $\rho$ living in $\mathbb{C}^{d}$ where $d=2^n$,\footnote{An $n$-qubit quantum state is a positive semi-definite operator on $\mathbb{C}^{2^n}$ such that $\Tr[\rho] = 1$.} output a $\hat{\rho}$ such that $\|\hat{\rho}-\rho\|_{tr}\leq \delta$ (where $\|\cdot \|_{tr}$ is the trace norm). Understanding the sample complexity of $\textsf{QST}$ has been a fundamental question in quantum information theory with applications in tasks such as verifying entanglement \cite{Kokail2021}, understanding correlations in quantum states \cite{cramer2010efficient}, and is useful for  understanding, calibrating and controlling noise in   quantum devices. A simple protocol for $\textsf{QST}$ uses  $T=O(d^6)$ copies: simply let $P_1,\ldots,P_{d^2}$ be all the $d$-dimensional Pauli matrices, use $O(d^2/\delta)$ copies of $\rho$ to estimate $\Tr(P_i\rho)$ up to error $\delta/d^2$. Using a technique of linear inversion, this is sufficient to produce $\hat{\rho}$ that satisfies  $\|\hat{\rho}-\rho\|_{tr}\leq \delta$. The overall sample complexity is $d^2\cdot O(d^4/\delta)=O(d^6/\delta)$. The dependence on the error $\delta$ is intuitive as more accurate description requires more measurements. Subsequently~\cite{flammia2012quantum} used techniques from compressive sensing to improve the complexity to $O(d^4/\delta^2)$ and after that Kueng et al.~\cite{kueng2017low} used more sophisticated techniques to improve the sample complexity to $O(d^3/\delta^2)$ and it was open for a while what was the right sample complexity of tomography. Two breakthrough works by Haah et al.~\cite{haah2017sample} and O'Donnell and Wright~\cite{o2016efficient} finally obtained optimal bounds for  the sample complexity of~$\textsf{QST}$. 

\begin{theorem}
\label{thm:tomsamcom}
The sample complexity of $\textsf{QST}$ up to trace distance $\delta$ is $O(d^2/\delta^2)$. Promised that the state is rank $r$, the sample complexity of $\textsf{QST}$ up to infidelity  $\varepsilon$ is $\tilde{\Theta}(dr/\varepsilon)$.\footnote{Infidelity between quantum states $\rho,\sigma$ is defined as $1-\|\sqrt{\rho}\sqrt{\sigma}\|_{tr}$.}
\end{theorem}
We now give a proof overview of a special case of this theorem - when the quantum state is pure (rank $r=1$). It makes use of the symmetric subspace and achieves a sample complexity of $\tilde{O}(d/\varepsilon)$. This is tight in $d$, as shown in~\cite{haah2017sample}.

\begin{proof}[Special case of Theorem \ref{thm:tomsamcom}]
Given an unknown $d$ dimensional pure state $\ket{\psi}^{\otimes k}$, with $k$ yet undetermined, note that the state lives inside the symmetric subspace $\Pi_{sym}^{d,k}$. To determine the state, one can perform  the so-called \emph{pretty-good measurement}~\cite{eldar2001quantum}, which has (continuous) POVM elements $\{\ketbra{\phi}{\phi}^{\otimes k}\}_{\ket{\phi}\in \mathbb{C}^d}$. Note that this measurement has infinitely many outcomes, which is ill defined, but we can address this by appropriate discretization. As a consequence, the measurement to be performed is 
$$
X \rightarrow {d+k-1\choose k} \int_{\phi}d\phi\hspace{1mm}\ketbra{\phi}{\phi}^{\otimes k}X\ketbra{\phi}{\phi}^{\otimes k}\otimes \ketbra{\text{description of }\phi}{\text{description of }\phi},
$$
which is a valid POVM whenever $X$ is in the symmetric subspace. The factor ${d+k-1\choose d}$ is the dimension of the symmetric subspace and ensures that the measurement is trace-preserving. Given $\ketbra{\psi}{\psi}^{\otimes k}$ as input, observe that a state $\ketbra{\phi}{\phi}$ is output with probability
$$
{d+k-1\choose k}\bra{\phi}^{\otimes k}\ketbra{\psi}{\psi}^{\otimes k}\ket{\phi}^{\otimes k}={d+k-1\choose k}|\braketIP{\phi}{\psi}|^{2k}.
$$
Thus, the probability that $|\braketIP{\phi}{\psi}|\leq 1-\varepsilon$ is at most
\begin{eqnarray*}
&&{d+n-1\choose n}\int_{\phi: |\braketIP{\phi}{\psi}|\leq 1-\varepsilon}d\phi|\braketIP{\phi}{\psi}|^{2k}\leq {d+k-1\choose k}\cdot (1-\varepsilon)^{2k}\leq \left(e\cdot\frac{k+d-1}{d}\right)^de^{-2k\varepsilon}.
\end{eqnarray*}
Choosing $k=\frac{10d}{\varepsilon}\log\frac{1}{\varepsilon}$, we can guarantee that RHS is small.
\end{proof}
In order to go from the special case to the theorem above, ~\cite{haah2017sample,o2016efficient} consider  a generalization of this argument and proceed by looking at subspaces that are invariant under permutations of registers and local unitary action. We refer the interested reader to~\cite{wright2016learn,oprimer} for a detailed exposition of the general proof. Very recently, the work of Flammia and O'Donnell~\cite{flammiaodonnell} considered the sample complexity of tomography under various distance metrics. A drawback of these tomography algorithms is that the time complexity of the procedure scales exponentially in $d$ (i.e., doubly-exponentially in $n$).  A natural question that was open from their work was, is there a \emph{time-efficient} procedure for tomography? In particular, is it possible to solve $\textsf{QST}$ using only single-copy measurements? There were a few works in this direction recently~\cite{yuen2022improved,lowe2022lower} and very recently Chen et al.~\cite{chen2022tight} answered this question with a surprisingly short proof.
\begin{theorem}
\label{thm:tomographylowerbound}
The sample complexity of $\textsf{QST}$ using single copy measurements is $\Theta(d^3/\delta^2)$.
\end{theorem}
The upper bound comes from the result of Kueng et al~\cite{kueng2017low} and  Chen et al.~\cite{chen2022tight} proved the lower bound of $\Omega(d^3/\delta^2)$ for $\textsf{QST}$ with single-copy (and, adaptive) measurements. We now sketch their lower bound. A technical challenge they had to overcome was the following: prior works that established sample lower bounds, proved this in the context of \emph{property testing}, where they  proved the hardness between distinguishing two hard distributions over states whose statistics (on separable measurements) were far apart. However, for tomography there are not too many techniques that we know to prove lower bounds against separable measurements. In this paper they use the so-called ``learning-tree framework" (which was first used in the prior work of Chen et al.~\cite{DBLP:conf/focs/ChenCH021} and inspired by classical decision trees which is used to analyze query complexity \cite{buhrman2002complexity}) to prove their lower bounds.  Here there is a tree where each node corresponds to a measurement operator applied onto a copy of the unknown state and the leaves are given by  classical bit string, corresponding to measurement labels. Based on the classical output in the leaves, the algorithm outputs an hypothesis state $\sigma$. The depth of the tree is the sample complexity of the learning algorithm. %In~\cite{chen2022tight}, the authors use the same tree-based argument that was introduced in a prior work (again for proving lower bounds against separable measurements). Here, the algorithm is given n copies of rho, and after a sequence of measurements, produces a classical bit string $x$, based on which the algorithm outputs the hypothesis state sigma. 
Chen et al.~\cite{chen2022tight} construct  a hard distribution of quantum states based on Gaussian ensemble matrices and   their main technical contribution is to show the following: if a separable tomography protocol is run on this hard instance, the leaves of the decision tree above (i.e., the output quantum state $\sigma$) is anti-concentrated around the unknown target quantum state if the depth of the tree is $o(d^3)$. Proving this anti-concentration is non-trivial, however the proof is fairly~short and we refer to their work for more.  

We conclude this section by discussing a simpler problem than $\textsf{QST}$:  quantum \emph{spectrum estimation}. Here the goal is to learn the spectrum of an unknown quantum state $\rho$, given copies of~$\rho$. It was showed~\cite{o2016efficient} that $O(d^2/\varepsilon^2)$ copies of~$\rho$ suffices to estimate the spectrum of~$\rho$ up to~$\ell_1$ distance $\varepsilon$ and they also showed a lower bound of $\Omega(d/\varepsilon^2)$.\footnote{They also showed that a \emph{class of algorithms} based on Schur sampling require an $\Omega(d^2/\varepsilon^2)$ sample complexity.} Spectrum learning has been an important subroutine in several property testing algorithms~\cite{o2015quantum,o2016efficient,wright2016learn,o2017efficient}. One question that remains open is the following:
\begin{question}
What is the tight sample complexity of quantum spectrum estimation?
\end{question}

\section{Learning physical quantum states}
\label{sec:physicalstates}
In the previous section we saw that fully learning arbitrary quantum states could require exponentially many copies of the unknown state. A natural question is, are there \emph{physical} subclasses of quantum states  which can be learned using polynomially many copies (and even polynomial time)? In this section, we discuss a few classes of physical states that can be learned using polynomial sample or time complexity. 

\subsection{Stabilizer states}
A natural candidate class that was considered for efficient tomography were  states that are known to be classically simulable. To this end, one of the first classes of states that were known to be learnable in polynomial time are stabilizer states. These are $n$-qubit states produced by the action of $n$-qubit Clifford circuits acting on $\ket{0^n}$. Aaronson and Gottesman~\cite{aaronson2004improved,aaronsongottesmantalk} considered this question and showed the following theorem. 
\begin{theorem}
\label{thm:stab}
The sample complexity of exactly learning $n$-qubit stabilizer states is $O(n)$ and the time complexity is $O(n^3)$.
\end{theorem}
In their paper,~\cite{aaronson2004improved} also showed that with single-copy measurements $O(n^2)$ copies of a stabilizer state $\ket{\psi}$ suffice to learn $\ket{\psi}$.  Subsequently, Montanaro~\cite{montanaro2017learning} gave a fairly simple procedure to learn stabilizer states using $O(n)$ copies that only uses entangled measurements over $2$ copies (prior to his work, Low~\cite{low2009learning} showed how to learn stabilizer states when one is allowed to make queries to the Clifford circuit preparing the unknown stabilizer state). We now discuss Montanaro's protocol: it is well-known~\cite{dehaene2003clifford,nest2008classical} that every $n$-qubit stabilizer state can be written as $\ket{\psi}=\frac{1}{\sqrt{|A|}}\sum_{x\in A}i^{\ell(x)}(-1)^{q(x)}\ket{x}$, where $A\subseteq \01^n$ is a subspace and $\ell$ (resp.~$q$) is a linear (resp.~quadratic) polynomial over $\FF_2$ in the variables $x_1,\ldots,x_n$. 

\begin{proof}[Special case of Theorem \ref{thm:stab}]
We consider the case when $\ell(x)=1$ for all $x$. Without loss of generality we can assume that $A=\01^n$ as well: a learning algorithm can measure $\tilde{O}(n)$ copies of $\ket{\psi}$ in the computational basis, learn the basis for $A$ and apply an  invertible transformation  to convert $\ket{\psi}$ to $\sum_{x\in \01^{k}\times 0^{n-k}}(-1)^{q(x)}\ket{x}$ where $\rank(A)=k$  and now apply a learning procedure on states of the form $\ket{\phi}=\frac{1}{\sqrt{2^k}}\sum_{x\in \01^k} (-1)^{q(x)}\ket{x}$. With this assumption, the learning algorithm uses the so-called Bell-sampling procedure: given two copies of $\ket{\phi_q}=\frac{1}{\sqrt{2^n}}\sum_{x} (-1)^{q(x)}\ket{x}$ where $q(x)=x^\top B x$ (where $B\in \FF_2^{n\times n}$), perform $n$ CNOTs between the first copy and  second copy, and measure the second copy. One obtains a uniformly random $y\in \FF_2^n$ and the~state 
$$
\frac{1}{\sqrt{2^n}}\sum_x (-1)^{f(x)+f(x+y)}\ket{x}=\frac{(-1)^{y^\top Ay}}{\sqrt{2^n}}\sum_x (-1)^{x^\top(B+B^\top)\cdot y}\ket{x}.
$$
The learning algorithm then applies the $n$-qubit Hadamard transform and measures to obtain bit string $(B+B^\top)\cdot y$. Repeating this process $O(n\log n)$ many times, one can learn $n$ linearly independent constraints about $B$. Using Gaussian elimination, allows one to learn the off-diagonal elements of $B$. {To learn the diagonal elements of $B$, a learner applies the operation $\ket{x}\rightarrow (-1)^{x_{ij}}\ket{x}$ if $B_{ij}=1$ for $i\neq j$. Repeating this for all $i\neq j$, the resulting quantum state is $\sum_x (-1)^{\sum_i x_i B_{ii}}\ket{x}$. Again applying the $n$-qubit Hadamard transform, the learner learns the diagonal elements of $B$.}  
\end{proof}

Given that stabilizer states are learnable using $O(n)$ copies, a followup question which hasn't received much attention is the following.

\begin{question}
    The stabilizer rank of $\ket{\psi}$ is the minimum $k$ for which $\ket{\psi}=\sum_i \alpha_i \ket{\phi_i}$ where $\ket{\phi_i}$ is an $n$-qubit stabilizer state. Can we learn stabilizer rank-$n$ states in polynomial time? 
\end{question}

Inspired by a result of Raz~\cite{DBLP:journals/jacm/Raz19} who proved time-space tradeoffs for parity learning, we also pose the following question.
\begin{question}
    The standard Bell-sampling approach for learning stabilizer states uses $O(n)$ copies of the stabilizer state and $O(n^2)$ classical space. If we have $o(n^2)$ classical space, what is the sample complexity of learning stabilizer states? Similarly, can we prove sample-space tradeoffs when the algorithm is given quantum space?\footnote{Recently, Liu et al.~\cite{liumemory} showed that an algorithm for learning parities needs either $\Omega(n^2)$ classical space, $\Omega(n)$ quantum space or $\Omega(2^n)$ labelled examples.}
\end{question}
    
\subsection{Learning circuits with non-Clifford gates}
% Last section we saw that  stabilizer states are learnable given copies of the state. Recently, Lai and Cheng~\cite{lai2022learning} consired
We saw how to learn the output states of Clifford circuits; a natural question is, if the circuit consists of a few \emph{non-Clifford} $T$ gates, can we still learn the output state?  It is known that that Clifford+$T$ circuits are universal for quantum computation and they have received much attention in fault-tolerance, circuit compilation and circuit simulation~\cite{BK98, FMM+12,KMM13, Sel15, RS16,BSS16, BG16,bravyi2019simulation}.  An arbitrary quantum circuit can be decomposed as a alternating sequence of Clifford stages and $T$ stages (by Clifford stage, we mean a Clifford circuit and by $T$-stage we mean either a $T$ gate or identity is applied to each qubit). The number of $T$ stages is  the \textit{$T$-depth} of the circuit. The learning task we consider is: Suppose $U$ is an $n$-qubit quantum circuit belonging to the class of $T$ depth-one circuits, can one learn $U$? In particular, if we are allowed to apply $U$ to specified prepared states and measure under a class of POVMs, how many measurements are required for learning $U$? In~\cite{lai2022learning}, they proved the following theorem.
\begin{theorem}
   Let $U$ be an $n$-qubit  $T$-depth one quantum circuit  comprising of $O(\log n)$ many $T$ gates. There exists a procedure that makes $\textsf{poly}(n)$ queries to $U$ and  outputs a circuit $\tilde{U}$  that is equivalent to $U$  when the input states are restricted to the computational~basis.
\end{theorem}

We omit the proof of this theorem and refer the reader to~\cite{lai2022learning} for more details.  Recently, there was a  hardness for learning the output distributions of  Clifford circuits with a \emph{single} $T$ gate~\cite{hinsche2022single},\footnote{The hardness result~\cite{hinsche2022single} is in a weaker \emph{statistical query} model which we discuss in Section~\ref{sec:sqlearning}, whereas the positive result~\cite{lai2022learning} considers the standard tomography model wherein the learner is given copies of the~state.} it is surprising that $T$-depth~$1$ circuits are learnable in polynomial time. This theorem naturally motivates the following questions.    

\begin{question}
    What is the complexity of learning circuits with $T$-depth $t$ (for some $t\geq 2$)?
    \end{question}
Recently, there have been a few works by Grewal et al.~\cite{grewal2023improvedlatest,grewal2023improved,grewal2022low} where they showed polynomial-time algorithms for learning states prepared by Clifford circuits  with~$O(\log n)$ many $T$ gates. They are also able to learn the output states of such circuits in polynomial time.\footnote{We remark that the states produced by these circuits have stabilizer rank $\leq n$, still leaving open the question we asked in the previous section.} 
     
    \begin{question}
    Can we learn $n$-qubit states and circuits that consist of $\omega(\log n)$ many $T$ gates? If not, is there a conditional hardness result one could show for learning these states?
\end{question}

\subsection{Learning phase states}
 In this section, we consider learning classical low-degree Boolean functions encoded as the amplitudes of quantum states, aka \emph{phase states}, which can be viewed as a generalization of stabilizer states. In recent times phase states have found several applications in cryptography, pseudo-randomness, measurement-based quantum computing, IQP circuits, learning theory~\cite{ji2018pseudorandom,brakerski2019pseudo,irani2021quantum,ananth2021cryptography,rossi2013quantum,takeuchi2019quantum,DBLP:journals/corr/abs-2208-07851}.

A degree-$d$ \emph{binary phase state} is a state of the form $\ket{\psi_f}=2^{-n/2}\sum_{x\in \01^n}(-1)^{f(x)}\ket{x}$ where $f:\01^n\rightarrow \01$ is a degree-$d$ function. Similarly, a degree-$d$ \emph{generalized phase state} is a state of the form   $\ket{\psi_f}=2^{-n/2}\sum_{x\in \01^n}\omega_q^{f(x)}\ket{x}$ where $f\, : \, \{0,1\}^n\to \ZZ_q$ is a degree-$d$ polynomial, $\omega_q=e^{2 \pi i/q}$ and $q$ is a prime. It is known that the output state of a random $n$-qubit Clifford circuit is a generalized $q=4$, degree-$2$ phase state with a constant probability~\cite{bravyi2016improved}, and 
a generalized degree-$d$ phase states with $q=2^d$ can be prepared from diagonal unitaries in the $d$-th level of the Clifford hierarchy~\cite{gottesman1999demonstrating,cui2017diagonal}. 
The learning question is: how many copies of $\ket{\psi_f}$ suffice to learn~$f$ exactly? Earlier works~\cite{bernstein1997quantum,montanaro2017learning,rotteler2009quantum}  showed $O(n)$ samples suffice for learning degree-$1$, and $O(n^2)$ suffices for learning degree-$2$ binary phase states; learning degree-$d$ for $d\geq 3$ has remained open (in fact it was plausible that it was a hard learning task since IQP circuits produce degree-$3$ phase states~\cite{montanaro2017circuits,bremner2011classical}). Sample complexity of learning \emph{generalized} phase states had not been studied before. 
% \begin{table}[h]
% \footnotesize
% \centering
%  \begin{tabular}{|c|c | c |} 
%  \hline
%   &  {Separable}  & Entangled\\ [0.5ex] 
%  \hline
%   {degree-$d$ Binary phase state } &  {$\Theta(n^{d})$} & {$\Theta(n^{d-1})$} \\[3mm] 
%  \hline
%  {degree-$d$ Generalized phase states } &  {$\Theta(n^{d})$} & {$\Theta(n^{d-1})$} \\[3mm]
%    \hline
%  \end{tabular}
%  \caption{Sample complexity of learning degree-$d$ binary and generalized phase states (for $d\geq 2$) with entangled and separable measurements.}
% \end{table}
In a recent work~\cite{DBLP:journals/corr/abs-2208-07851} they provided separable and entangled bounds for learning phase~states. Below we sketch the upper and lower bounds for the case of separable measurements. We refer the reader to~\cite{DBLP:journals/corr/abs-2208-07851} for the proof of the sample complexity with entangled measurements.

\textbf{Separable measurements} \emph{upper bound.} The proof makes the following simple observation: given $\ket{\psi_f}=2^{-n/2}\sum_x \omega_q^{f(x)}\ket{x}$, suppose we measure qubits $2,3,\ldots,n$ in the computational basis and  obtain~$y \in \{0,1\}^{n-1}$. The post-measurement state is then  
$
\ket{\psi_{f,y}} = ( \omega_q^{f(0y)} \ket{0} + \omega_q^{f(1y)} \ket{1})/\sqrt{2}.
$ 
If the base of the exponent was $(-1)$, then applying a Hadamard on $\ket{\psi_{f,y}}$ produces $c=f(0y)-f(1y)$. Their main idea is, it is still possible to obtain a value $b\in \ZZ_q$ such that $b\neq c$ with certainty. To this end, consider a POVM whose elements are given by $\M=\{\ketbra{\phi_b}{\phi_b}\}_{b\in \ZZ_q}$,
where $\ket{\phi_b}=( \ket{0} - \omega_q^{b} \ket{1})/\sqrt{2}$.
 Applying this POVM $\M$ onto an unknown state $(|0\ra+\omega_q^c|1\ra)/\sqrt{2}$ they observe that $c$ is the outcome with probability $0$ and furthermore one can show that \emph{every} other outcome $b\ne c$ appears with  probability $\Omega(d^{-3})$. Repeating this process $m=O(n^{d-1})$ many times, one obtains $(y^{(k)},b^{(k)})$ for $k=1,2,\ldots,m$ such that $f(1y^{(k)})-f(0y^{(k)})\neq b^{(k)}$ for all $k \in [m]$. Let $g(y)=f(1y^{(k)})-f(0y^{(k)})$ (i.e., $g=\nabla_1 f$). Clearly $g$ is a degree $\leq d-1$ polynomial. A non-trivial analysis in~\cite{DBLP:journals/corr/abs-2208-07851} shows the following:  the probability of having more than one polynomial 
degree-$d-1$ polynomial $g$ satisfying the constraints $g(y^k)\neq b^k$ is exponentially small if we choose $k=\tilde{O}(q^3 n^{d-1})$. Hence $k$ many copies $\ket{\psi_f}$, allows a learning algorithm to learn the derivative $\nabla_1 f$. Repeating this for $n$ directions, we can learn $\nabla_1 f,\ldots,\nabla_n f$, hence $f$.

\textbf{Separable measurements} \emph{lower bound.} Furthermore, they show that the above protocol is optimal even if allowed  single \emph{copy} measurements. The main  idea is the following: for a uniformly random degree-$d$ function $f$, suppose a learning algorithm measures the phase state $\ket{\psi_f}$ in an arbitrary orthonormal basis $\{U\ket{x}\}_x$. One can  show that the distribution describing the measurement outcome $x$ is ``fairly" uniform. In particular, 
$    \mathop{\mathbb{E}}_f [H(x|f)]\ge n - O(1), 
$ where $H(x|f)$ is the Shannon entropy of
a distribution $P(x|f)=|\la x|U^*|\psi_f\ra|^2$. To prove this, they first lower bound the Shannon entropy by Renyi-two entropy and prove a technical statement to bound the latter by deriving  an explicit formula for $\EE_f [|\psi_f\ra\la \psi_f|^{\otimes 2}]$. Thus, for a typical $f$, measuring one copy of the phase state $\ket{\psi_f}$ provides at most $O(1)$ bits of information about $f$. Since a random uniform
degree-$d$ polynomial $f$ with $n$ variables has entropy~$\Omega(n^d)$, one has to measure $\Omega(n^d)$ 
copies of $\ket{\psi_f}$ in order to learn $f$.

% \textbf{Entangled measurements.} The entangled 

% Our learning procedure is based on the pretty good measurement~(PGM). To prove our bound, we show the following: in order to learn degree-$d$ phase state, the \emph{optimal} measurement \emph{is} the PGM since the ensemble $\Sh=\{\ket{\psi_f}\}_f$ is \emph{geometrically uniform} (GU), i.e., the ensemble is $\Sh=\{U_f\ket{\phi}\}_{f}$ where $\{U_f\}_f$ is an Abelian group. 
% Next, we show that if $\Sh$ is a GU ensemble, then the success probability of the PGM in correctly identifying $f$ is \emph{independent} of $f$; finally, we prove a new technical tool generalizing earlier works of~\cite{ben2012random,beame2020bias}: with probability $\geq 1- \exp(-c_2 \binom{n}{\leq d})$, a uniformly random $f:\FF_2^n\rightarrow \ZZ_{2^d}$ and $j\in \ZZ_{2^d}^*$, satisfies $|\Exp_x[\omega^{j\cdot f(x)}]| \leq  2^{-c_1 n/d}$.  We then probabilistically construct an ensemble $\Sh'\subseteq \Sh$ for which $\langle \psi_f|\psi_g\rangle \leq 2^{-n/d}$ for all $f,g\in \Sh'$ and also $|\Sh'|=\Omega(2^{n^d})$. Using a well-known result of PGMs, we show an overall upper bound of $O(dn^{d-1})$ for sample complexity of learning phase states using entangled~measurements. \snote{mention entangled lower bound}

In~\cite{DBLP:journals/corr/abs-2208-07851} they also constructed a procedure to learn circuits (consisting of diagonal gates in the Clifford hierarchy) which produce phase states, leaving open the following:

\begin{question}
  What is the complexity of learning circuits consisting of \emph{non-diagonal} gates in the Clifford hierarchy?\footnote{When given query access to the circuit, Low~\cite{low2009learning} gave a procedure to learn the Clifford hierarchy.}
\end{question}
To this end, when given query access to the circuit, Low~\cite{low2009learning} gave a procedure to learn the Clifford hierarchy. However given only copies of $C\ket{0^n}$ where $C$ consists of \emph{non-diagonal} gates in the Clifford hierarchy, the question we ask is open.  Liang~\cite{liang2022clifford} recently showed a conditional hardness of learning Clifford circuits in the \emph{proper learning} setting. An open question from~\cite{DBLP:journals/corr/abs-2208-07851} which might improve our understanding of phase states is,  how many copies suffice to \emph{test} phase states.

\begin{question}
What is the complexity of  property testing degree-$d$ phase states? In particular, given copies of a state $\ket{\psi}$ promised it is either a degree-$d$ phase state or $\varepsilon$-far from the set of all degree-$d$ phase states, how many copies are necessary and sufficient to distinguish these cases?
\end{question}

%\cite{cramer2010efficient}
% \subsection{Trivial states}

% [The work with Yunchao is super slow here and wont be published anytime soon...maybe lets skip this section?]

\subsection{Gibbs states of local Hamiltonians}
In this section we discuss the problem of learning a Hamiltonian given copies of its Gibbs state. The setup of this learning problem is as follows: let $H$ be a local Hamiltonian $H=\sum_{\alpha=1}^m\mu_\alpha E_\alpha$ on $n$ qubits, where $E_\alpha$ is some local orthogonal operator basis such as the Pauli matrices, an algorithm receives copies of the Gibbs state $\rho_\beta(H) = \frac{e^{-\beta H}}{\Tr(e^{-\beta H})}$ and the goal is to  output a list of numbers $\mu':=\{\mu'_1, \mu'_2, \ldots ,\mu'_m\}$ that are close to $\mu:=\{\mu_1, \mu_2, \ldots ,\mu_m\}$ in either the $\ell_\infty$ or $\ell_2$ distance metric. We make the natural assumption that $\mu_1,\ldots ,\mu_m \in (-1,1)$, which simply says that each local interaction has bounded strength. Learning an unknown Hamiltonian from its Gibbs state has been studied in statistical physics and machine learning \cite{ChowLiu1968, hinton1986learning, Tanaka1998,Albert201499} for many decades,  known as the ``inverse Ising~problem". 

For machine learning, one is often interested in Ising interaction (that is, each $E_\alpha$ is a Pauli operator of the form $Z\otimes Z$) where the underlying interaction graph\footnote{An interaction graph has the qubits in the Hamiltonian as its vertices and each Ising interaction as its edge.} is sparse and unknown \cite{Bresler_learning, Interaction_screening, Klivans_learning}. Learning the Ising model also learns the very important underlying graph structure. In the quantum regime, we are far from being able to learn the underlying graph, solely under the sparsity assumption. Thus, we will assume that the underlying graph is known. For most physics applications, the graph can also respect the geometric constraints that arise from living in a low dimensional space. Before discussing algorithms for Hamiltonian learning, we first discuss motivation for considering this learning question. 

Hamiltonian learning can be a useful experimental tool in a variety of settings.
\begin{itemize}
    \item \textbf{Understanding the lattice structure:} Suppose we wish to know whether interactions in a given quantum material respect a Kagome lattice structure or a square lattice structure, assuming one of them is the case. This knowledge can significantly affect the physical properties, such as the electronic behaviour as looked at by~\cite{Jiang21}. If our Hamiltonian learning algorithm guarantees that $\|\mu'-\mu\|_\infty\leq \frac{1}{3}$, then we can figure out which edge is present or absent, in turn the lattice structure.
    \item 
    \textbf{Estimating the spectral gap of a Hamiltonian:} Another key quantity of interest is the spectral gap of a Hamiltonian, which dictates a myriad of ground state properties. For learning the spectral gap up to  constant precision error (say $0.1$), we need to know the Hamiltonian really well and the right regime to consider is $\|\mu'-\mu\|_1\leq 0.1$.  
    \item 
    \textbf{Effective Hamiltonians:} Local Hamiltonians are - after all - models of real interactions happening in physics. Effective Hamiltonians regularly arise when we wish to consider interactions between a specific set of particles or quasi-particles. These interactions can be hard to precisely determine theoretically, motivating the use of Hamiltonian learning \cite{SHBSGB22}. 
    \item 
    \textbf{Entanglement Hamiltonian:} Li-Haldane conjecture states that the marginals of a 2D gapped ground state are Gibbs state of a local Hamiltonian with temperature that depends on the location of the local term. Learning this Hamiltonian is directly relevant to understanding the entanglement structure of the system \cite{Kokail2021}. 
\end{itemize}

We remark that in the applications above, we did not specify the inverse temperature $\beta$ of the Gibbs state. In some cases, the temperature can be controlled, and then setting $\beta$ to be a small constant leads to optimal algorithms - see below. In other cases, such as for effective or entanglement Hamiltonians, temperature can be very low at the boundary of the region. Thus, efficient algorithms for Hamiltonian learning at all finite temperatures has interesting consequences in quantum computing.

\subsubsection{Sufficient statistics}
We now turn to designing algorithms for the learning task above. One natural question is: given an instance of Hamiltonian learning problem, is there any data about the Gibbs state that would suffice to learn the Hamiltonian? In other words, what are the `\emph{sufficient statistics}' for the Hamiltonian? The answer turns out to be very simple: they are the set of expectation values $f_\alpha= \Tr(E_\alpha \cdot \rho_\beta(H))$. 

There are two ways to prove that sufficient statistics suffice for learning
\begin{itemize}
    \item \textbf{Information theoretic argument:} Let's consider two Gibbs quantum states $\rho_\beta(H)$ and $\rho_\beta(G)$, where $H=\sum_\alpha \mu_\alpha E_\alpha$ and $G=\sum_\alpha \nu_\alpha E_\alpha$. We will argue that their ``distance'' is characterized by the expectation values. For this, we evaluate the symmetric relative entropy
\begin{eqnarray*}
\relent{\rho_\beta(H)}{\rho_\beta(G)}+\relent{\rho_\beta(G)}{\rho_\beta(H)}&=&\beta\Tr((H-G)(\rho_\beta(G)-\rho_\beta(H)))\\
&=& \beta\sum_\alpha(\mu_\alpha-\nu_\alpha)\Tr(E_\alpha(\rho_\beta(G)-\rho_\beta(H))),
\end{eqnarray*}
where the first equality follows by routine calculation. Thus, if $\Tr(E_\alpha \rho_\beta(H))=\Tr(E_\alpha \rho_\beta(G))$ for all $\alpha,$ the right hand side vanishes. The above argument also says that if $\Tr(E_\alpha \rho_\beta(H))\approx\Tr(E_\alpha \rho_\beta(G))$ then the relative entropy between the Gibbs quantum states is small. This is good enough to `learn' the Gibbs state up to small error in total variational distance. More~precisely, 
\begin{align*}
&\relent{\rho_\beta(H)}{\rho_\beta(G)}+\relent{\rho_\beta(G)}{\rho_\beta(H)} \\
&\leq \beta \max_\alpha|\mu_\alpha-\nu_\alpha|\cdot\left(\sum_\alpha|\Tr(E_\alpha \rho_\beta(H))-\Tr(E_\alpha \rho_\beta(G))|\right)\\
&\leq  2\beta \left(\sum_\alpha|\Tr(E_\alpha \rho_\beta(H))-\Tr(E_\alpha \rho_\beta(G))|\right).
\end{align*}
However, this estimate is not sufficient to guarantee the closeness of the Hamiltonians.
\item \textbf{Convexity of log partition function:} A more useful argument - for our problem description - is based on the convexity of the log partition function. The observation here is simply that the function $\log \Tr(e^{-\beta H})$ is a convex function in the parameters $\{\mu_1,\mu_2,\ldots, \mu_m\}$. The vector $(f_1,f_2,\ldots, f_m)$ of trace expectations then forms the gradient of this function. Furthermore, precise knowledge of the gradient can be used to identify the parameters $\mu_1,\ldots,\mu_m$. See Figure \ref{fig:gradientpf} (a).
\end{itemize}

\begin{figure}
\centering
\begin{subfigure}[b]{0.4\textwidth}
\centering
\begin{tikzpicture}[xscale=1, yscale=1]
%\draw[black, thick] (0,0) rectangle (5,4);
\draw [thick, ->] (0.5,0.5) -- (0.5, 3.5);
\node at (0,3) {$\log Z$};
\draw [thick, ->] (0.5,0.5) -- (4.5, 0.5);
\node at (4,0) {$\mu$};
\draw [thick] (1,2) to [out=290, in=180] (2,1.5) to [out=0, in =260] (4, 3);
\draw [->] (2,1.17) -- (4,2.17);
\node at (3,2) {$\mu^*$};
\end{tikzpicture}
\caption{}
\end{subfigure}
\hspace{1cm}
\begin{subfigure}[b]{0.4\textwidth}
\centering
\begin{tikzpicture}[xscale=1, yscale=1]
%\draw[black, thick] (0,0) rectangle (5,4);
\draw [thick, ->] (0.5,0.5) -- (0.5, 3.5);
\node at (0,3) {$\log Z$};
\draw [thick, ->] (0.5,0.5) -- (4.5, 0.5);
\node at (4,0) {$\mu$};
\draw [thick] (1,2) to [out=290, in=180] (2,1.5) to [out=0, in =260] (4, 3);
\draw [->] (2,1.17) -- (4,2.17);
\node at (3,2) {$\mu^*$};
\draw [red, ->] (3,1.17) -- (4.2,2.97);
\node at (3.7,2.6) {\textcolor{red}{$\mu'$}};
\end{tikzpicture}
\caption{}
\end{subfigure}
\caption{(a) Given the gradient of a convex function - such as the log partition function - there is a unique point that matches the gradient. (b) Strong convexity ensures that good knowledge of the gradient leads to good enough closeness to the desired point.}
\label{fig:gradientpf}
\end{figure}
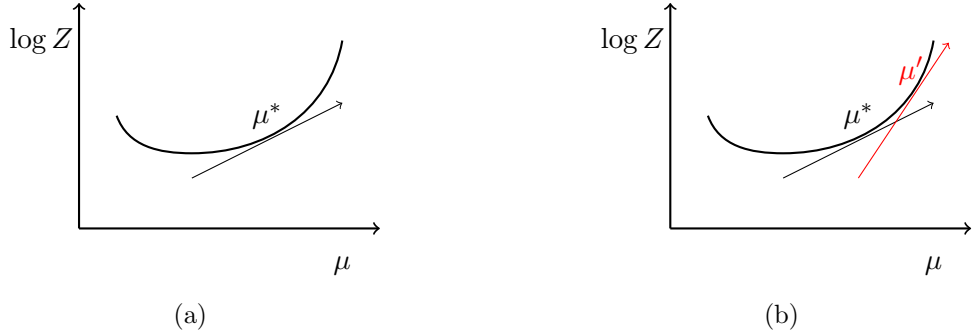

The quantities $f_\alpha$ can only be known approximately in experiments, due to statistical errors in estimation. Thus, a robust version of sufficient statistics is needed to develop an  algorithm for Hamiltonian learning. In \cite{anshu2021sample}, strong convexity of the log partition function was established. This roughly says that the log partition function ``curves well'' (see Figure \ref{fig:gradientpf} (b)). An algorithm - based on gradient descent - was constructed which uses  $O\left(m^3\cdot 1/\varepsilon^2\cdot \poly(1/\beta)\cdot {\exp(\poly(\beta)})\right)$ copies of the Gibbs state to learn the Hamiltonian with guarantee $\|\mu'-\mu\|_2\leq \varepsilon$. The time complexity of the algorithm depends on computing the gradient of the partition function. An efficient computation at high temperatures, for stoquastic Hamiltonians and 1D Hamiltonians - but requiring large run-time for low temperatures and arbitrary Hamiltonians. 

\subsubsection{Commuting Hamiltonians}
While the above algorithm based on sufficient statistics takes exponential time at low temperatures, classical Hamiltonians can be learned time-efficiently using more refined techniques - as noted earlier \cite{Bresler_learning, Interaction_screening, Klivans_learning}.  In fact, here we argue that \emph{commuting} Hamiltonians - that include classical Hamiltonians - can also be efficiently learned at any temperature, as long as the interaction graph is known. The algorithm is fundamentally different from the previous one that was based on estimating the expectation values $f_\alpha=\Tr(E_\alpha \rho_\beta(H))$. Consider $H=\sum_\ell h_\ell$, where the commutator $[h_\ell, h_{\ell'}]=0$. We note that this notation is different from the one we used earlier, in particular here the $h_\ell$s need not be an \textit{orthogonal} basis. The algorithm, sketched in \cite{AAKScomm21} is based on the following theorem. See also Figure \ref{fig:commeff}

\begin{theorem}
\label{thm:commeffec}
\cite{AAKScomm21} For any region $R$ on the lattice, define the effective reduced Hamiltonian $H_R= \frac{-1}{\beta}\log \tr_{R^c}\left(\rho_\beta\right)$.\footnote{Here the subscript in $\Tr$ refers to the registers being traced out. Moreover, $R^c$ is the set of qubits not in $R$.} Let $\partial R$ be the boundary of $R$, and $\partial_{-} R$ be the inner boundary of $R$ (which is the set of qubits in $R$ that interact with a qubit outside $R$). Then
$$
H_R= \alpha_RI + h_R + \Phi,
$$
where $\Phi$ is only supported on $\partial_{-} R$ and $[\Phi, h_R]=0$. Here, $\alpha_R$ is some real number and $\|\Phi\|\leq 2|\partial R|$.
\end{theorem}

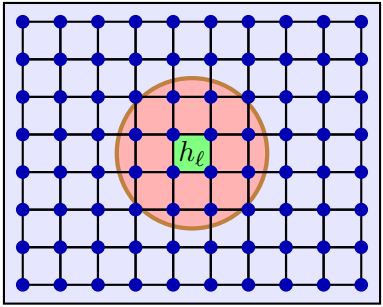
\begin{figure}
\centering
\begin{tikzpicture}[xscale=0.5,yscale=0.5]
\draw [fill=blue!10!white, thick] (0.5,0.5) rectangle (10.5, 8.5);

\draw [ultra thick, brown, fill=red!30!white] (5.5,4.5) circle [radius=2];

\foreach \i in {1,...,5}
{
\foreach \j in {1,...,4}
\draw [thick] (2*\i-1,2*\j-1) rectangle (2*\i,2*\j);
}

\foreach \i in {1,...,4}
{
\foreach \j in {1,...,3}
\draw [thick] (2*\i,2*\j) rectangle (2*\i+1,2*\j+1);
}

\foreach \i in {1,...,5}
{
\foreach \j in {1,...,3}
\draw [thick] (2*\i-1,2*\j) rectangle (2*\i,2*\j+1);
}

\draw [thick, fill=green!50!white] (5,4) rectangle (6,5);

\node at (5.5,4.5) {$h_\ell$};

\foreach \i in {1,...,4}
{
\foreach \j in {1,...,4}
\draw [thick] (2*\i,2*\j-1) rectangle (2*\i+1,2*\j);
}

\foreach \i in {1,...,10}
{
\foreach \j in {1,...,8}
   \draw (\i, \j) node[circle, fill=black!30!blue, scale=0.5]{};
}
\end{tikzpicture}
\caption{Consider the marginal of the Gibbs state in the brown circle. For commuting Hamiltonians, this marginal is the Gibbs state of a Hamiltonian that is the boundary correction to the Hamiltonian strictly within the region.}
\label{fig:commeff}
\end{figure}

Using this theorem, the learning algorithm is straightforward: perform good enough tomography of the region around an interaction $h_\ell$ to reconstruct the marginal to very high accuracy. Then take log of the marginal, followed by computing each $h_\ell$ up to error $\varepsilon$ (i.e., output a $h'_\ell$ such that $\|h'_\ell-h_\ell\|\leq\varepsilon$). This is good enough to estimate the unknown Hamiltonian $H$. The resulting sample complexity~\cite{AAKScomm21} is
$\exp({\mathcal{O}(\beta k^D)})\cdot \mathcal{O}\left(1/{\varepsilon^2}\cdot \log (m/\delta)\right),
$
where $k$ is the locality of the Hamiltonian, $D$ is the degree of the underlying interaction graph and $\delta$ is the probability of failure. Time complexity is 
$m\cdot {\exp({\mathcal{O}(\beta k^D)}})\cdot \mathcal{O}\left(1/{\varepsilon^2}\cdot \log (m/\delta)\right)$.

\subsubsection{High temperature Gibbs states}
The idea of using effective reduced Hamiltonian in Theorem~\ref{thm:commeffec} can also be applied to non-commuting Hamiltonians, as long as the temperature is high enough (or $\beta$ smaller than the critical temperature~$\beta_c$, which is a constant). This follows from a similar result as Theorem \ref{thm:commeffec} shown by \cite{KKB20} using cluster expansion, with $H_R$ approximated by $h_\ell+\Phi$ as $\|H_R-h_{\ell}-\Phi\|_\infty\leq \exp(-\Omega(r))$ for a spherical region $R$ of radius $r$ around $\ell$ (see Figure \ref{fig:noncommeff} for an example).

\begin{figure}
\centering
\begin{tikzpicture}[xscale=0.5,yscale=0.5]
\draw [fill=blue!10!white, thick] (0.5,0.5) rectangle (10.5, 8.5);

\draw [ultra thick, brown, fill=red!30!white] (5.5,4.5) circle [radius=3];

\draw [<->] (2.5, 0) -- (8.5, 0);
\node at (5.5, -0.4) {\small $2r$};

\foreach \i in {1,...,5}
{
\foreach \j in {1,...,4}
\draw [thick] (2*\i-1,2*\j-1) rectangle (2*\i,2*\j);
}

\foreach \i in {1,...,4}
{
\foreach \j in {1,...,3}
\draw [thick] (2*\i,2*\j) rectangle (2*\i+1,2*\j+1);
}

\foreach \i in {1,...,5}
{
\foreach \j in {1,...,3}
\draw [thick] (2*\i-1,2*\j) rectangle (2*\i,2*\j+1);
}

\draw [thick, fill=green!50!white] (5,4) rectangle (6,5);

\node at (5.5,4.5) {$h_\ell$};

\foreach \i in {1,...,4}
{
\foreach \j in {1,...,4}
\draw [thick] (2*\i,2*\j-1) rectangle (2*\i+1,2*\j);
}

\foreach \i in {1,...,10}
{
\foreach \j in {1,...,8}
   \draw (\i, \j) node[circle, fill=black!30!blue, scale=0.5]{};
}
\end{tikzpicture}
\caption{In the high temperature regime, it has been shown by \cite{KKB20} that the marginal of a Gibbs state is still the Gibbs state of the original Hamiltonian (within the brown circle of radius $r$) up to  boundary correction. However, the boundary term has some support within the circle and has strength $\approx e^{-r}$ near the center. Thus, to learn $h_\ell$, we need to make sure that $e^{-r}$ is small~enough.}
\label{fig:noncommeff}
\end{figure}
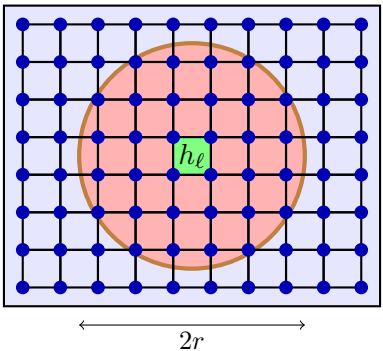
If we use the same approach as above, to estimate each $h_\ell$ with error $\varepsilon$, we thus need $r=O(\log 1/\varepsilon)$. The sample complexity now incurs an additional factor of $\exp(r^D)$, where $D$ is the lattice dimension or the degree of the graph. Thus, the sample complexity is
% $$\frac{e^{\mathcal{O}(\beta (kr)^D)}}{\varepsilon^2}\log\frac{m}{\delta}=\frac{e^{\mathcal{O}(\beta (\log\frac{1}{\varepsilon})^D)}}{\varepsilon^2}\log\frac{m}{\delta},$$
$O\Big(\exp({\beta (\log\frac{1}{\varepsilon})^D})\cdot 1/{\varepsilon^2}\cdot \log ({m}/{\delta})\Big)$  
and time complexity is 
$O\Big(m\cdot \exp({(\log\frac{1}{\varepsilon})^D})\cdot 1/{\varepsilon^2}\cdot \log({m}/{\delta})\Big)$. For constant $\varepsilon$, this is very efficient; however for $\varepsilon=\frac{1}{m}$, in which case the $\ell_1$ error of learning is small enough, the sample complexity is super-polynomial in $m$. This is somewhat unsatisfactory, as this approach seems worse than \cite{anshu2021sample} in the regime where each local term has to be learned very accurately. In a subsequent work~\cite{haah2021optimal} provided a unifying and complete answer for $\beta<\beta_c$. Employing the cluster expansion method of \cite{KKB20, KS20}, the authors directly express the sufficient statistics $\Tr(E_\alpha \rho_{\beta}(H))$ as an infinite series in $\beta$ with coefficients polynomial in the local Hamiltonian terms. Approximate knowledge of the sufficient statistics is then inverted to estimate the Hamiltonian terms. They achieve tight sample and time complexity of
$\mathcal{O}\left(1/\varepsilon^2 \cdot \log ({m}/{\delta})\right)$ {and} $\mathcal{O}\left(m/\varepsilon^2\cdot \log ({m}/{\delta})\right)$ respectively.\footnote{We note that there were gaps in the above results - that originated in \cite{KS20} - were fixed in \cite{WildA22}. See \cite{haah2021optimal} for a discussion.}

\subsubsection{Discussion}

The problem of time efficient Hamiltonian learning - on a fixed geometry and at arbitrary $\beta$ - remains open. The fact that this is possible in the commuting case is encouraging, as there is no prior reason to expect that the commuting and non-commuting cases would be fundamentally different. Indeed, very good heuristic methods exist for the task \cite{Aradlearning, Qi2019learningFromGroundState}. We end this section with two relevant open questions.

\begin{question}
    Can we achieve Hamiltonian learning under the assumption that the Gibbs states satisfy an approximate conditional independence?\footnote{Given a quantum state $\rho$ on registers $A,B,C$, $\mathrm{I}(A:C|B)_{\rho} = S(\rho_{AB})+S(\rho_{BC})- S(\rho_B)-S(\rho_{ABC})$ is the quantum conditional mutual information. We say that $\rho$ satisfies approximate conditional independence if $\mathrm{I}(A:C|B)_{\rho}\approx 0$.} 
\end{question}

Approximate conditional independence is known to hold in 1D \cite{KatoB2019} and conjectured to hold for every dimension. 

\begin{question}
Pseudorandomness is a bottleneck for learnability. If a family of quantum states are pseudo-random, then polynomially many copies of the state are indistinguishable from Haar random states by any efficient quantum algorithm. Could low temperature Gibbs states be pseudorandom, which would explain the difficulty in finding time efficient algorithm?
\end{question}

\subsection{Matrix product states}
Matrix Product States (MPS) are a widely used representation of quantum states on a spin-chain. Mathematically, a state $\ket{\psi}\in (\mathbb{C}^{d})^{\otimes n}$ is a \emph{matrix product state} (MPS) if $\ket{\psi}$ can be written as
	$$
	\ket{\psi}=\sum_{i_1,\ldots,i_n \in [d]} \Tr(\MPS{A}{1}{i_1} \cdot \MPS{A}{2}{i_2}\cdots \MPS{A}{n}{i_n}) \, \ket{i_1,\ldots,i_n},
	$$
	where for all $j \in [n], i \in [d]$, $\MPS{A}{j}{i}$ is a  $D_j\times D_{j+1}$ matrix. We call the set of matrices $\{\MPS{A}{j}{i}\}$ an \emph{\textsf{MPS} representation of $\ket{\psi}$}. We refer to $D=\max_{j} D_j$ as the \emph{bond dimension} of $\ket{\psi}$, when minimized over all \textsf{MPS} representations. Many physically relevant $n$-qubit quantum states - such as gapped ground states - can be approximated by \textsf{MPS} with bond dimension polynomial in $n$. A learning algorithm for an \textsf{MPS} state $\rho$ takes as input, copies of $\rho$ promised to be an \textsf{MPS} of certain bond dimension $D$ and outputs an \textsf{MPS} of bond dimension $D'$ that approximates $\rho$ in fidelity. The goal is learn these states with polynomial sample and time complexity  along with keeping $D'$ close to~$D$.

\begin{figure}[h]
\centering
\begin{tikzpicture}[xscale=0.5,yscale=0.5]

\draw [ultra thick] (0,4.5) -- (16,4.5);

\foreach \j in {0,...,4}
{
\draw [ultra thick, blue] (4*\j,4.5) -- (4*\j,6);
\draw [ultra thick, brown, fill=red!30!white] (4*\j,4.5) circle [radius=1];
\node at (4*\j,4.5) {$A_{\j}$};
}
\end{tikzpicture}
\caption{A matrix product state is specified by its bond dimension $D$ and a sequence of $D\times D$ matrices.  The virtual bonds (black lines) indicate the amount of entanglement and the physical blue lines represent qudits.}
\label{fig:MPS}
\end{figure}
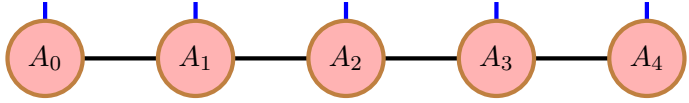
Unlike Gibbs states, local observable statistics do not always determine an MPS. For example, consider the CAT states $\frac{1}{\sqrt{2}}\ket{00\cdots 0} \pm \frac{1}{\sqrt{2}}\ket{11\cdots 1}$, which are \textsf{MPS} of bond dimension 2. These states can't be distinguished on any set of $n-1$ qubits. Thus, any algorithm for \textsf{MPS} must make global measurements. Indeed, \cite{LLP10, CPFSGBLPL10} gave a polynomial time algorithm to learn an MPS, using global-but-efficient measurements. Their algorithm learns a sequential circuit that prepares the MPS. The resulting output has bond dimension $D'=\text{poly}(D)$. 

\begin{figure}[h]
\centering
\begin{tikzpicture}[xscale=0.5,yscale=0.5]

\draw [ultra thick] (0,4.5) -- (16,4.5);

\foreach \j in {0,...,4}
{
\draw [ultra thick, blue] (4*\j,4.5) -- (4*\j,6);
\draw [ultra thick, brown, fill=red!30!white] (4*\j,4.5) circle [radius=1];
\node at (4*\j,4.5) {$A_{\j}$};
}

\draw [blue] (-2,3) rectangle (10,6.5); 
\draw [blue] (2,2.5) rectangle (14,6); 

\end{tikzpicture}
\caption{Under coarse graining (blue rectangles) the physical dimension exceeds the bond dimension. For example, if $D=4$ and each physical blue line is a qutrit, the physical dimension of blue regions is $3^3=27$, which is larger than the total bond dimension at the boundary $4^2=16$. For typical tensors $A_1, \ldots, A_n$, this makes the map from  the virtual bonds to physical qudits invertible. Such an \textsf{MPS} is injective.}
\label{fig:MPSinj}
\end{figure}
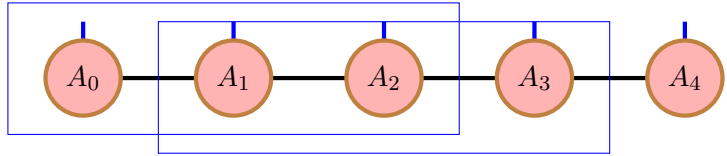

However, local measurements are more ideal in the experimental settings. Under the assumption of injectivity (see Figure~\ref{fig:MPSinj}) learning an \textsf{MPS} with just local measurements may be possible. Let us first observe that a `local' sufficient statistics holds for injective \textsf{MPS}. This is because the marginals on $O(\log D)$ qudits determine the parent Hamiltonian which has the \textsf{MPS} as its unique ground state. Barring the statistical errors - which can be addressed using injectivity - the knowledge of parent Hamiltonian allows one to reconstruct an approximation to the \textsf{MPS} state using the rigorous algorithm in \cite{LandauVV15, AradLVV17}. A drawback of this approach is that the algorithm in \cite{LandauVV15, AradLVV17} outputs an \textsf{MPS} with bond dimension $D'=\text{poly}(n)$, which may be much larger than a constant $D$.  Thus, we see a large blowup in bond dimension of the output MPS.   Cramer et al.~\cite{CPFSGBLPL10} proposes a heuristic efficient algorithm based on the singular value thresholding algorithm in which the output bond dimension does not suffer such blow-up; however there is no guarantee that the output \textsf{MPS} is  close to the input MPS. This leads us to the following~question:
\begin{question}
    Can an injective \textsf{MPS} be learned efficiently using local measurements, with the output bond dimension $D'=\poly(D)$? 
\end{question}
A possible direction is to improve~\cite{LandauVV15,AradLVV17} under suitable~guarantees.
\begin{question}
 Promised that the ground state is an \textsf{MPS} of bond dimension $D$, can the algorithm in \cite{LandauVV15,AradLVV17} be improved to produce an output \textsf{MPS} that also has $\poly(D)$ bond dimension? 
\end{question}
Projected Entangled Pair States (\textsf{PEPS}) are higher dimensional generalizations of Matrix Product States. Learnability of \textsf{PEPS} is far from clear, even in terms of sample complexity. Observe that the parent Hamiltonians of \textsf{PEPS} are frustration-free and locally-gapped. A natural question is, can the learning task become simpler assuming injectivity?
\begin{question}
    Given an injective \textsf{PEPS}, an approximation to the parent Hamiltonian can be learned with polynomial sample and time complexity. Can this be used to write down a description of another \textsf{PEPS} that represents a similar state? 
\end{question}
The key bottleneck above is that there is no two dimensional analogue of \cite{LandauVV15}, despite an area law for locally-gapped frustration-free spin systems~\cite{AAG22}. 
\section{Alternate models of learning quantum states}
\label{sec:alternatemodels}
In order to perform full state tomography on $n$ qubits we saw that it is necessary and sufficient to obtain $\Theta(2^{2n})$ many copies of the unknown state. The exponential scaling of the complexity is prohibitive for experimental demonstrations. of course a natural question is, is it  necessary to approximate the unknown state up to  small trace distance? In particular, do there exist weaker but still practically useful learning goals, with smaller sample complexity? These questions have led some to consider  learning only the `useful' properties of a unknown quantum state. There have been several models of learning quantum states (inspired by computational learning theory) where exponential savings in sample complexity is possible and we discuss these models in this section.

\subsection{$\PAC$ learning and online learning}

\paragraph{$\PAC$ learning.} In a seminal work, Valiant~\cite{DBLP:journals/cacm/Valiant84} introduced the \emph{Probably Approximately Correct} ($\PAC$)  model of learning which lays the foundation for computational learning theory.  In this model, there is a concept class consisting of Boolean functions  $\Cc\subseteq \{c:\01^n\rightarrow \01\}$ and an underlying distribution $D:\01^n\rightarrow [0,1]$. The learning algorithm is provided with labelled examples of the form $(x,c(x))$ where $x$ is sampled from the distribution $D$.\footnote{We assume that the concept class is Boolean here, one could also consider real-valued classes where $c(x)$ is then specified up to certain bits of precision.} We say a learning algorithm $(\varepsilon,\delta)$-learns a concept class $\Cc$ if it satisfies the~following:
\begin{quote}
    For every $c\in \Cc$, distribution $D:\01^n\rightarrow [0,1]$, given labelled examples $(x,c(x))$ where $x$ is sampled from $D$: with probability at least $1-\delta$, the algorithm outputs $h:\01^n\rightarrow \pmset{}$ such that $\Pr_{x\sim D} [h(x)= c(x)]\geq 1-\varepsilon$.
\end{quote}
The sample complexity of a learning algorithm $\A$ is the maximum number of labelled examples,  over all the concepts $c\in \Cc$ and distributions $D$.   The $(\varepsilon,\delta)$-sample complexity of a concept class $\Cc$ is the minimum sample complexity over all $(\varepsilon,\delta)$-$\PAC$ learners $\A$ for $\Cc$. Similarly, one can define the sample complexity (resp.~time complexity)  of $(\varepsilon,\delta)$-learning $\Cc$  as the samples used (resp.~time taken) by the $(\varepsilon,\delta)$-learning~algorithm. There have been many works in classical literature that have looked at \emph{distribution-dependent} $\PAC$ models wherein the distribution $D$ is known to the learner and the algorithm needs to perform well under~$D$.

Aaronson~\cite{aaronson2007learnability} considered the natural analog of learning quantum states in the $\PAC$ model. In this model of learning, the  concept class $\Cc$ is a collection of functionals described by an unknown quantum states, $\rho \in \Cc$ acting on the class of measurements operators $\calE$ and $D:\calE\rightarrow [0,1]$ is an unknown distribution over all possible 2-outcome measurements.   A quantum learning algorithm obtains several examples of the form $(E_i,\Tr(\rho E_i))$ where $E_i$ is drawn from the distribution $D$ and the goal is to approximate $\rho$. We say a learning algorithm $(\varepsilon,\delta,\gamma)$-learns $\Cc$ if it satisfies the~following: 
\begin{quote}
    For every $\rho\in \Cc$, given examples $(E_i,\Tr(\rho E_i))$ where $E_i\sim D$, with probability at least $1-\delta$, output $\sigma$ that satisfies $\Pr_{E\sim D}[|\Tr(\rho E)-\Tr(\sigma E)|\leq \varepsilon]\geq 1-\gamma$.
\end{quote}
%The sample complexity of a concept class $\Cc$ is defined similar to the classical $\PAC$ sample complexity.  
%The sample complexity of a learning algorithm $\A$ is the maximum number of examples,  over all the states $\rho\in \Cc$ and distributions $D$.  The $(\varepsilon,\delta)$-sample complexity of a concept class $\Cc$ is the minimum sample complexity over all $(\varepsilon,\delta)$-$\PAC$ learners $\A$ for $\Cc$. 
In contrast to tomography where the output state $\sigma$ is close to the unknown $\rho$ in trace distance, i.e., $\sigma$ should satisfy $\max_{E}|\Tr(E\rho)-\Tr(E\sigma)|\leq \varepsilon$, in $\PAC$ learning the goal is for $\Tr(E\rho)$ to be close to $\Tr(E\sigma)$ for \emph{most $E$s}. In a surprising result, Aaronson showed that the class of all $n$-qubit quantum states can be $\PAC$-learned using just  $O(n)$ samples.
\begin{theorem}
The sample complexity of $\PAC$ learning $n$-qubit quantum states is $O(n\cdot \textsf{poly}(1/\varepsilon,1/\delta,1/\gamma))$. 
\end{theorem}
Similarly, Cheng et al.~\cite{cheng2015learnability} considered the ``dual problem" of learning a quantum measurement in the $\PAC$ learning framework. We do not prove these theorems here, we refer the reader to the survey~\cite[Theorem~4.16]{DBLP:journals/sigact/ArunachalamW17}.  A natural question left open by Aaronson was, what classes are states are \emph{time-efficiently} $\PAC$ learnable? To this end, Rocchetto~\cite{rocchetto2017stabiliser} observed that the class of stabilizer states is $\PAC$ learnable in polynomial time. The learning algorithm of Rocchetto assumed that the distribution $D$  was over Pauli observables and he crucially used that $\Tr(P \rho)\in \{-1,1,0\}$ when $\rho$ was a stabilizer state. This allowed Rocchetto to learn the stabilizers of the unknown $\rho$ and with some extra work, the entire stabilizer state $\rho$.  A natural question that remains open is the following.
\begin{question}
What is the time complexity of $\PAC$ learning states prepared by Clifford circuits with $t$ many T gates? What is the $\PAC$ sample complexity of learning stabilizer-rank $k$ states?
\end{question}
 Gollakota and  Liang~\cite{gollakota2022hardness} considered a natural question of learning stabilizer states in the presence of noise. They looked at a restrictive version of $\PAC$ learning, called \emph{statistical query} learning (we discuss this model in further detail in Section~\ref{sec:sqlearning}). Here the learning algorithm is allowed to make single-copy ``queries" to  learn the unknown \emph{noisy} stabilizer state (the noise model they consider is the depolarizing noise). In this model,~\cite{gollakota2022hardness} showed that learning stabilizer states with noise is as hard as learning parities with noise (LPN) using classical samples (which is believed to require exponentially many samples~\cite{blum2003noise}).

\paragraph{Online learning.} Subsequently, Aaronson et al.~\cite{aaronson2018online}, Chen et al.~\cite{chen2022adaptive}  looked at the setting of online learning quantum states (inspired by the classical model of online learning  functions). The online model can be viewed as a variant of tomography and $\PAC$ learning. Consider the setting of tomography, suppose it is infeasible to possess $T$-fold tensor copies of a quantum state $\rho$, but instead we can obtain only sequential copies of $\rho$. The quantum online learning model consists of repeating the following rounds of interaction: the learner obtains a copy of $\rho$ and a description of measurement operator $E_i$ (possibly adversarially) and uses it to predict the value of $\Tr(\rho E_i)$. In the $i$th round, if the learners prediction was $\alpha_i$ and $\alpha_i$ satisfies $|\Tr(\rho E_i)- \alpha_i|\leq \varepsilon$ then it is correct, otherwise it has made a mistake. The goal of the learner is the following: minimize $m$ so that after making $m$ mistakes (not necessarily consecutively), it makes a correct prediction on \emph{all} future rounds. Aaronson~\cite{aaronson2018online} showed that it suffices to let $m$ be 
  the \emph{sequential fat-shattering dimension} of~$\calC$, denoted $\sfat(\Cc)$ (a combinatorial parameter  introduced in~\cite{rakhlin2015online} to understand classical online learning), which in turn can be upper bounded by $O(n/\varepsilon^2)$ for the class of $n$-qubit quantum states.

\subsection{Shadow tomography}
A caveat of the quantum $\PAC$ learning model is that the learning algorithm has to only perform well under a distribution and it is a priori unclear if the $\PAC$ model is a natural model of learning.  Aaronson~\cite{aaronson:shadow} introduced another learning model called \emph{shadow tomography}. Here, the goal of a learning algorithm algorithm is as follows: let $E_1,\ldots,E_m$ be positive semi-definite operators satisfying $\|E_i\|\leq 1$, how many copies of an $n$-qubit state $\rho$ are necessary and sufficient in order to estimate $\Tr(\rho E_1),\ldots,\Tr(\rho E_m)$ up to additive error $\varepsilon$. There are two naive protocols for this task: $(i)$ either do quantum state tomography which takes $\exp(n)$ many copies and allows to estimate $\Tr(\rho E_i)$ for all $i$, or $(ii)$ take $O(m/\varepsilon^2)$ many copies of $\rho$ and estimate up to error $\varepsilon$ each of the $\Tr(\rho E_i)$s. Surprisingly, Aaronson showed that one can perform the task of shadow tomography exponentially better in both $m$ and~$n$  in \emph{sample complexity}, but still running  in time exponential~in~$n$.

\begin{theorem}
There is a protocol for shadow tomography that succeeds with probability $\geq 2/3$ using $\widetilde{O}((n\log^4 m)/\varepsilon^4)$ many copies of $\rho$.
\end{theorem}
We now sketch a proof of this theorem.  For simplicity, we let $\varepsilon$ be a constant, say $1/3$. Aaronson's proof is based on the technique of post selected learning~\cite{aaronson2007learnability} which was introduced in the context of communication complexity. In this communication task, there are two players Alice and Bob: Alice has a $d$-dimensional quantum state $\rho$ (unknown to Bob) and together they know a set of $m$ many operators $\{E_1,\ldots,E_m\}$. The goal is for Alice to send a classical message to Bob, who should output $\Tr(\rho E_1),\ldots,\Tr(\rho E_m)$ up to  error $1/3$.  The same two trivial protocols we mentioned earlier would work here, giving a communication upper bound of $O(m+d^2)$. 
%There are two trivial protocols for this: (a) Alice sends the entire description of $\rho$ which takes $O(d^2)$ bits, (b) Alice sends Bob an $\varepsilon$-approximation of  $\Tr(\rho E_i)$ for all $i\in [k]$, which takes communication $O(k)$. 
Surprisingly, Aaronson~\cite{aaronson2007learnability} showed that there exists a communication protocol with cost  $\textsf{poly}(\log d,\log m)$ which solves the communication task, whose proof we sketch first.  Bob starts by guessing the state Alice possesses. To this end, he lets $\rho_0=\id/d$, the maximally mixed state, and updates his guess in every round. At the $t$th round, suppose Bob’s guess is $\rho_t$ (whose classical description is known to Alice), Alice communicates to Bob a $j\in [m]$ for which $|\Tr(\rho E_j)-\Tr(\rho_t E_j)|$ is the largest and sends him $b = \Tr(E_j\rho)$. With this, Bob updates $\rho_t\rightarrow\rho_{t+1}$ as follows: let $q = O(\log \log d)$ and
$F_t$ be a two-outcome measurement on $\rho_t^{\otimes q}$ that applies the POVM $\{E_j, \id-E_j\}$ to each of the $q$ copies of $\rho_t$
and accepts if and only if the number of $1$-outcomes was at least $(b - 1/3)q$.  Suppose $\sigma_{t+1}$ is the state obtained by post-selecting on $F_t$ accepting $\rho^{\otimes q}_t$, then $\rho_{t+1}$
is the state obtained by tracing out the last $q-1$ registers of $\sigma_{t+1}$.  Aaronson
showed that after $T = O(\log d)$ rounds, Bob will have $\rho'$ which satisfies $|\Tr(E_i\rho) - \Tr(E_i\rho')| \leq 1/3$ for $i \in [m]$. Returning to shadow tomography,  observe that there is no Alice, and Bob is replaced by a quantum learner. So, at the $t$th stage, without any
assistance, the learner needs to figure out $j\in [m]$ for which $|\Tr(E_j\rho_t)-\Tr(E_j\rho)|$ is
large. To this end, Aaronson  used a variant of the Quantum
OR lemma~\cite{DBLP:conf/soda/HarrowLM17}, which uses $O(\log m)$ copies of $\rho$ and outputs ``yes" if there exists a $j\in [m]$ for which $|\Tr(E_j \rho) - \Tr(E_j\rho)| \geq 2/3$ and outputs
``no" if $|\Tr(E_j \rho) - \Tr(E_j\rho)| \leq 1/3$ for every $j\in [m]$. However, in order to use the ideas from the communication protocol,  in the “yes” instance of the OR lemma, Bob needs to know $j$ (not just the existence of $j$) in order to update $\rho_t$ to $\rho_{t+1}$. Aaronson shows how to do this by using a simple binary search over $\{E_1,\ldots, E_m\}$ to find such a $j$. Putting these ideas together, Aaronson shows the sample complexity upper bound for the shadow tomography.

\subsection{Max-entropy principle and Matrix Multiplicative Weight Update}
Recall that to solve shadow tomography, the goal is to find a quantum state $\sigma$ that satisfies $\tr(\sigma E_i)\approx \tr(\rho E_i)$ for all $i$. Further, one would like to minimize the number of copies of $\rho$, suggesting that $\sigma$ should be no more informative than matching the above expectations. This is an ideal ground to invoke the \emph{max entropy principle}, which states that the quantum state $\sigma$ maximizing $S(\sigma)$ (maximum uncertainty) subject to the constraints $\tr(\sigma E_i)= \tr(\rho E_i)$ is the Gibbs quantum state $\frac{e^{-\sum_i \alpha_i E_i}}{\tr(e^{-\sum_i \alpha_i E_i})}$. Here, $\alpha_i$s are determined by the expectations $\tr(\rho E_i)$. From here, an algorithm for shadow tomography can start with a trivial guess for $\sigma$ - the maximally mixed state - which is then updated as new knowledge from $\rho$ arrives. Since the maximally mixed state is the Gibbs state of the trivial Hamiltonian `$0$', the updates can be done directly to the Hamiltonian. To see how this update can be determined, consider a technical theorem from \cite{FBK21}.
\begin{theorem}
    Consider a Hamiltonian $G$ and an operator $E$ with $\|E\|_\infty\leq 1$. For $\eta\in \R$, consider the Gibbs states $\sigma=\frac{e^{-\beta H}}{\tr(e^{-\beta H})}$ and $\sigma'=\frac{e^{-\beta (H+\eta E)}}{\tr(e^{-\beta (H+\eta E)})}$. It holds that for any quantum state~$\rho$, 
$$
\relent{\rho}{\sigma'} - \relent{\rho}{\sigma}\leq \beta\cdot \eta\Big(\beta\eta e^{|\beta\eta|} + \tr(E(\rho-\sigma))\Big).
$$ 
In particular, setting $\eta= -\frac{\tr(E(\rho-\sigma))}{4\beta}$, we find that  
$$
\relent{\rho}{\sigma'} - \relent{\rho}{\sigma}\leq -\tr(P(\rho-\sigma))^2/8.
$$
\end{theorem}
\begin{proof}
To prove this result, a direct calculation reveals that $$\relent{\rho}{\sigma'} - \relent{\rho}{\sigma}=\beta\eta\tr(\rho E)+\log\frac{\tr(e^{-\beta (H+\eta E)})}{\tr(e^{-\beta H})}.$$ Using the Golden-Thompson inequality, we find that 
$$\relent{\rho}{\sigma'} - \relent{\rho}{\sigma}\leq\beta\eta\tr(\rho E)+\log\frac{\tr(e^{-\beta H}e^{-\beta\eta E})}{\tr(e^{-\beta H})}=\beta\eta\tr(\rho E)+\log\tr(\sigma e^{-\beta\eta E}).$$
Since $\|E\|_\infty \leq 1$, we can estimate $\tr(\sigma e^{-\beta\eta E})\leq 1-\beta\eta\tr(\sigma E) + \beta^2\eta^2 e^{|\beta\eta|}$, which implies 
$$\relent{\rho}{\sigma'} - \relent{\rho}{\sigma}\leq\beta\eta\tr(\rho P)+\log(1-\beta\eta\tr(\sigma E) + \beta^2\eta^2 e^{|\beta\eta|})\leq \beta\eta\tr((\rho-\sigma) P) + \beta^2\eta^2e^{|\beta\eta|}.
$$
This proves the theorem statement.  
\end{proof}
Thus, the alternate algorithm for shadow tomography proceeds by identifying an $E_i$ that still does not satisfy $\tr(E_i\rho)=\tr(E_i\sigma)$ and then updating the weight of such an $E_i$ in $\sigma$. In order to find such an $E_i$ with $\text{poly}(\log m)$ sample complexity, one can use the `quantum OR lemma' as described earlier. We also highlight that this procedure can be used to learn the Hamiltonian. Assuming that $\rho$ itself is a Gibbs state, we update the weights of the basis operators $E_\alpha$ until the expectation values are close. In such a case, the strong convexity from \cite{anshu2021sample} ensures that the Hamiltonian is learned up to  desired error.

\subsection{Subsequent works building on shadow tomography}
\label{sec:subsequentshadow}

There have been several subsequent works that have built upon Aaronson's shadow tomography protocol which we discuss in this section.

\subsubsection{Classical shadows} 
A subsequent work of Huang, Kueng and Preskill~\cite{huangpreskill} presented an alternate protocol for a restricted version of shadow tomography that is more time efficient than Aaronson's original shadow tomography protocol. In particular, they proved the following.
\begin{theorem}
Let $B>0$ be an integer and $\varepsilon,\delta \in [0,1]$. Given $O(B/\varepsilon^2 \log(1/\delta))$ copies of $\rho$,  there exists a procedure that satisfies the following: for every observable $M$  that satisfies  $\Tr(M^2)\leq B$, with probability $\geq 1-\delta$, the quantity $\Tr(\rho M)$ can be computed to error $\varepsilon$.
\end{theorem}
To compare this procedure and shadow tomography, suppose the algorithm needs to estimate $m$ many expectation values, then by letting $\delta\sim 1/m$ with success probability $\geq 2/3$, the overall sample complexity scales as $O((\log m) \cdot B/\varepsilon^2)$. Additionally, observe that the procedure above is \emph{independent} of the observables $M$, unlike Aaronson's protocol~\cite{aaronson:shadow} which used the observables $E_1,\ldots, E_m$ in a crucial way to learn $\Tr(\rho E_i)$. We now give a proof sketch of the theorem: they first give a  polynomial-time procedure for generating \emph{classical shadows} of the unknown quantum state $\rho$ using $T=O(B/\varepsilon^2 \log(1/\delta))$ copies of $\rho$. These classical shadows are generated by running the following procedure: given copies of $\rho$, the algorithm samples a uniformly random Clifford $C$, computes $C\rho C^\dagger$ and measures the state in the computational basis to get an $n$-bit string $b$. So the classical shadows is the set $\{(C_i,b_i\}_{i\in [T]}$. Using these classical shadows,~\cite{huangpreskill} use a simple median of means estimation procedure to estimate $\Tr(\rho M)$ for an arbitrary observable $M$ satisfying~$\Tr(M^2)\leq B$. Thus the sample complexity is $O(B/\varepsilon^2\log(1/\delta))$.

\subsubsection{Improved shadow tomography and agnostic learning} 
 B{\u a}descu and O'Donnell~\cite{DBLP:conf/stoc/BadescuO21} improved the complexity of shadow tomography to $\tilde{O}((n\cdot \log^2 m)/\varepsilon^2)$, which simultaneously obtains the best known dependence on each of the parameters $n,m,\varepsilon$. We do not sketch their protocol, but remark on one interesting corollary of shadow tomography is a procedure which they call \emph{quantum hypothesis selection}. Although not phrased in this language, quantum hypothesis selection can be viewed as \emph{agnostic} learning quantum states. The setup for quantum agnostic learning states is the following: $\Cc$ is a collection of known quantum states $\{\rho_1,\ldots,\rho_m\}$, a learning algorithm is provided with copies of an unknown state $\sigma$ and needs to find $\rho_{k} \in \Cc$ which is closest to $\sigma$ in the following sense: output $\rho_{k}\in \Cc$ such that
 \begin{equation}
 \label{eq:agnostic}
 \|\rho_{k}-\sigma\|_1\leq \alpha \cdot \min_{\rho\in \Cc}  \|\rho-\sigma\|_1+\varepsilon,
 \end{equation}
 for some $\alpha$. We briefly sketch the reduction from quantum agnostic learning  to shadow tomography:  consider the two states $\rho_i,\rho_j$ in the concept class $\Cc$, by Holevo-Helstrom's theorem there exists an optimal measurement measurement $\{A_{ij},\id-A_{ij}\}$ such that $\Tr(A_{ij}\cdot (\rho_i-\rho_j))=\|\rho_i-\rho_j\|_{tr}$. Now perform shadow tomography using $\tilde{O}((n\cdot \log^2 m)/\varepsilon^2)$ copies of  $\sigma$ along with the operators $\{A_{ij}\}_{i,j\in [m]}$ to obtain $\alpha_{ij}$s satisfying  $|\alpha_{ij}-\Tr(A_{ij}\sigma)|\leq \varepsilon/2$. At this 
 point,~\cite{DBLP:conf/stoc/BadescuO21} simply goes over all $\rho \in \Cc$ to find a $\rho_{k}$ that minimizes the quantity $\max_{i,j}|\Tr(\rho_{k}A_{ij})-\alpha_{ij}|$ (this is inspired by  classical hypothesis selection~\cite{yatracos1985rates}). Let $\eta=\min_{\rho\in \Cc}  \|\rho-\sigma\|_1$ and $i^*=\argmin_{\rho\in \Cc}  \|\rho-\sigma\|_1$. Observe that
 \begin{align*}
\|\rho_{k}-\sigma\|_{tr}&\leq \eta+\|\rho_{k}-\rho_{i^*}\|_{tr}\\
&=\eta+|\Tr(A_{i^*k}\rho_k)-\Tr(A_{i^*k}\rho_{i^*})|\\
&\leq \eta+|\Tr(A_{i^*k}\rho_k)-\alpha_{i^*k}|+|\Tr(A_{i^*k}\rho_{i^*})-\alpha_{i^*k}|\leq 3\eta+\varepsilon.
 \end{align*}
 Hence the resulting $\rho_{k}$ satisfies Eq.~\eqref{eq:agnostic} with $\alpha=3$.   As far as we are aware,~\cite{DBLP:conf/stoc/BadescuO21,chung2018sample,fanizza2022learning,caro2021binary} are the only few works to look at agnostic learning of quantum states. 
 %Their main result is that, one can do quantum agnostic learning (for $\alpha=3$) using  samples of $\sigma$.  
 These works gives  rise to the following two interesting~questions.

\begin{question}
 What is the sample complexity of quantum agnostic learning if we require $\alpha=1$?\footnote{In classical computational learning theory, reducing the value of $\alpha$ to $1$ has been resolved for certain Boolean function classes in the seminal works~\cite{DBLP:journals/jacm/LinialMN93,DBLP:conf/stoc/GopalanKK08}.}   \end{question}

\begin{question}
What classes of states can be agnostic learned time-efficiently? Can we learn stabilizer states efficiently in the quantum agnostic model?
%\end{itemize}
 
\end{question}
 
\subsubsection{Shadow tomography with separable measurements} 
Chen et al.~\cite{DBLP:conf/focs/ChenCH021} considered the problem of shadow tomography if one was only allowed \emph{separable} measurements. In this setting, they showed that $\tilde{\Omega}(\min\{m,d\})$ many copies are necessary for shadow tomography,  matching the upper bound of Huang et al.~\cite{huangpreskill} of $\tilde{O}(\min\{m,d\})$. They in fact show that, in order to estimate the 
expectation values of all $4^n$ many $n$-qubit Pauli observables, one needs $\Omega(2^n)$ copies of $\rho$ (given access to only separable measurements). The proof of this lower bound follows the following three-step approach $(i)$ 
They first consider the learning tree framework that we discussed below Theorem~\ref{thm:tomographylowerbound}, wherein there is a tree with each node corresponding to a measurement applied to the unknown state $\rho$ and the leaves of the tree correspond to the $m$ many expectation values.
%They introduce a new lower bounding technique called the \emph{tree-representation}\snote{no need to emph this, write better}:
$(ii)$ Using this learning tree technique, the main technical lemma they show is that, in order to prove the hardness of estimating $\Tr(\rho Q_i)$ for arbitrary $Q_i$, using separable measurements, it suffices to upper bound $\delta(Q_1,\ldots,Q_{2^n})=\frac{1}{m}\sup_{\ket{\psi}}\sum_i \langle \psi|Q_i|\psi\rangle^2$. $(iii)$ Finally they show that for the Paulis $P_i$, we have that $\delta(P_1,\ldots,P_{2^n})$ is exactly $1/({2^n}+1)$, which immediately gives them their sample complexity lower bound of $\Omega(2^n)$. 

Additionally they also consider settings wherein the learning algorithm is adaptive (i.e., the learner can perform measurements based out of previous measurement outcomes) and non-adaptive (i.e., the learning algorithm needs to decide at the beginning a sequence of measurements to carry out). Similarly, a followup work of Gong and Aaronson~\cite{gong2022learning} showed how to perform shadow tomography when given $m$ many $k$ outcome measurements using $\poly(k,\log m,n,1/\varepsilon)$ copies of $\rho$.

\subsection{Equivalence between quantum learning models}

So far, we saw many many seemingly (unrelated) models of computation aimed at learning an unknown quantum state such as, shadow tomography, $\PAC$ learning, communication complexity, online learning.  Aaronson and Rothblum~\cite{aaronson2019gentle} also considered differential privacy in learning quantum states and used this notion to prove new bounds on online learning and shadow tomography.\footnote{Classically differential privacy was formalized in the seminal works by Dwork~\cite{dwork2006differential,dwork2014algorithmic}: we say a learning algorithm is differentially private if it behaves approximately the same when given two training datasets which differ in only limited number of entries.} A natural question is, is there a connection between these models? In~\cite{DBLP:conf/nips/QuekAS21} they showed ``equivalences" between all these models of computation.  A high-level overview of the results in their work is summarized in the figure below. For technical reasons, we do not discuss pure and approximate $\DP$ in detail,  we simple remark that \emph{pure} $\DP$ is a stronger requirement than \emph{approximate} $\DP$ and refer the reader to~\cite{DBLP:conf/nips/QuekAS21} for more details. In particular, these equivalences imply that algorithms in one framework gives rise to quantum learning algorithms in other frameworks.
%Fundamental to these connections is the $\sfat$ dimension, which characterizes the sample complexities of these learning tasks. In the process, we also develop a new upper bound on the $\sfat$ dimension in terms of a maximized Holevo information. 
%This is an important first step towards porting these quantum learning models, originally developed for learning qubit states, over to continuous-variable quantum information processing. On a meta level, 
%This connection gives the Shannon-theoretic notions of channel capacity and Holevo information, a learning-theoretic interpretation, and our results can be seen as an important interdisciplinary bridge between these~fields. 
 We remark that only a few of these arrows are efficient in both sample and time complexity, otherwise these implications are primarily information-theoretic. 

\begin{figure}[!ht]
\centering
\begin{tikzpicture}
[->,>=stealth',shorten >=1pt,auto,  thick,yscale=0.8,
main node/.style={circle,draw}, node distance = 0.8cm and 1.8cm,
block/.style   ={rectangle, draw, text width=5em, text centered, rounded corners, minimum height=2.5em, fill=white, align=center, font={\footnotesize}, inner sep=5pt}]
    %     \node[main node,block,] (PureDP$\PAC$) at (-4,2) {\textsf{Pure\\ DP $\PAC$}};
    % \node[main node,block] (Prdim) at (-0.5,2) {\textsf{Representation dimension}};
    \node[main node,block] (Roneway) at (3,2) {\textsf{Pure\\ DP $\PAC$}};
   \node[main node,block] (sfat) at (6,0.5) {\textsf{Sequential fat-shattering}};
  \node[main node,block] (online) at (3,-1.2) {\textsf{Online learning}};
  \node[main node,block] (stability) at (-0.5,-1.2) {\textsf{Stability}};
  \node[main node,block] (Appdim) at (-4,-1.2) {\textsf{Approximate DP $\PAC$}};
    \node[main node,block] (shadow) at (9,-1.2) {\textsf{Shadow tomography}};
    % \path [->] (PureDP$\PAC$) edge node {} (Prdim);
    % \path [->](Prdim) edge node {} (Roneway);
    \path [->](Roneway) edge node {} (sfat);
    \path [->](sfat) edge node {\emph{1}} (online);
    \path [->](sfat) edge node {} (shadow);
    \path [->](online) edge node {} (shadow);
    \path [->](online) edge node {\emph{2}} (stability);
    \path [->](stability) edge node {\emph{3}} (Appdim);
\end{tikzpicture}
\end{figure}
Although  {\em a priori}, it seems that $\PAC$ learning, online learning and $\DP$ learning have little to do with one another, classically there have been a sequence of works establishing tight connections between these three fields~\cite{DBLP:journals/siamcomp/KasiviswanathanLNRS11}. The main center piece in establishing these connections is the notion of stability which was introduced in a recent breakthrough work of Bun et al.~\cite{DBLP:conf/focs/BunLM20}. In~\cite{DBLP:conf/nips/QuekAS21} they ``quantize" these connections. Below we give a sketch of their proofs and refer to their work for a detailed overview. 

It is well-known classically that if there is a $\DP$ $\PAC$ learning algorithm for a class $\calC$ then the \emph{representation dimension} of the class is small. Representation dimension then upper-bounds classical communication complexity and $\sfat(\calC)$. In~\cite{DBLP:conf/nips/QuekAS21} they show that this connection carries over in a simple way to the quantum setting.

 \emph{\textbf{(1)}}: Let $\calC$ be a concept class of states with  finite $\sfat(\calC)$. In order to describe an online-learner for $\calC$ making at most $\sfat(\calC)$ mistakes, in~\cite{DBLP:conf/nips/QuekAS21} they construct a \emph{robust} standard optimal algorithm (denoted $\RSOA$) whose accuracy guarantees are robust to adversarial imprecision in the training feedback. The $\RSOA$ algorithm is inspired by the classical standard optimal algorithm for Boolean functions (however in the quantum setting it needs to work for real functions as well as with adversarial noise).  Aaronson et  al.~\cite{aaronson2018online} showed an upper bound of $\sfat(\calC)\leq n$ on the number of mistakes in this setting asking if there is an \emph{explicit} algorithm that achieves this bound (their $\RSOA$ made explicit this algorithm). 

\emph{\textbf{(2)}}: Here, they show that a concept class $\calC$ with $\sfat(\calC)=d$  can be learned by a stable algorithm. To prove this, they follow the technique of~\cite{DBLP:conf/focs/BunLM20} which feeds a standard optimal algorithm (which they replace with $\RSOA$) with a specially-tailored input sample. The tailoring algorithm deliberately injects ``mistake" examples into the sample, each of will force a prediction mistake in $\RSOA$. Since the $\RSOA$ {\em completely} identifies the target concept after making at most $d$ prediction mistakes, the injection step allows the stable algorithm to output the correct hypothesis. In~\cite{DBLP:conf/nips/QuekAS21}, their quantum-focused adaptation of this technique handles the twin challenges of accurately engineering the mistake examples for real-valued functions, and having $\varepsilon$-uncertainty in the adversary's feedback (both of which are not present in the Boolean setting).

  \emph{\textbf{(3)}}:  Now that one has a stable algorithm established above, one need to make it differentially private. In the Boolean setting, given a stable algorithm $\mathcal{A}$, there is a well-known ``Stable Histograms" algorithm may be used a `wrapper' around $\mathcal{A}$, to privately identify $\mathcal{A}$'s high-probability output functions. This involves running $\mathcal{A}$ many times and outputting its most frequent output, while adding Laplacian noise to make it $\DP$. However, they encounter an additional complication in the quantum setting: outputting the ``most frequent" quantum state doesn't make sense, since two quantum states could be arbitrarily close and qualify as valid outputs. Addressing this,~\cite{DBLP:conf/nips/QuekAS21} modify Stable Histograms and show that it can be used to make the quantum stable learning algorithm $\DP$.

% The main center piece in establishing these connections is \emph{quantum stability}, which is also our new conceptual contribution. We say a learning algorithm~$\A$ is $(T,\varepsilon,\eta)$-\emph{stable} if, given $T$ many labelled examples $S=\{(E_i,\Tr(\rho E_i))\}$, there exists a state $\sigma$ such~that
% $$
% \Pr[\A(S) \in \calB(\varepsilon, \sigma)]\geq \eta,
% $$
% where $\B(\varepsilon,\sigma)$ is the ball of states that  $\varepsilon$-agree with $\sigma$ on most measurement outcomes and the probability above is taken over the randomness of $\calA$. Intuitively, quantum stability means that up to  an $\varepsilon$-distance, there is some $\sigma$ that is output by $\A$ with ``high" (at least $\eta$) probability. Classically, Dwork and Roth~\cite{dwork2014} state that ``\emph{Differential privacy is enabled by stability and ensures stability...we observe a tantalizing moral equivalence between learnability, differential privacy, and
% stability,}" and this was crucially~used in~\cite{bun2020equivalence,alon2019private,abernethy2019online,bousquet2019passing}. Such a connection has remained unexplored in the quantum setting and ours is the first work that explores this interplay between privacy, stability and quantum~learning. 

%\subsection{Learning quantum noise}

\section{Learning classical functions through quantum encoding}
\label{sec:classicalfunctionlearning}
\subsection{Learning Boolean functions}
The \emph{quantum} $\PAC$ model for learning a concept class of Boolean functions $\Cc\subseteq \{c:\01^n\rightarrow \01\}$ was introduced by Bshouty and Jackson~\cite{DBLP:conf/colt/BshoutyJ95}. In this model, instead of access to labelled examples $(x,c(x))$ where $x$ is sampled from $D$, the quantum learning algorithm is provided with copies of the \emph{quantum example} state $\ket{\psi_c}=\sum_{x\in \01^n}\sqrt{D(x)}\ket{x,c(x)}$. Quantum examples are a natural generalization of classical labelled examples (by measuring a single quantum example, we obtain a classical labelled example). A quantum $\PAC$ learner is given  copies of the quantum example state,  performs a POVM (where each outcome of the POVM is associated with an hypothesis) and outputs the resulting hypothesis.  The sample complexity of the learner here is measured as the number of \emph{copies} of $\ket{\psi_c}$ and the $(\varepsilon,\delta)$-quantum sample complexity of learning $\Cc$ is defined similarly to the classical $\PAC$ learning.  There have been a few works that have looked at quantum $\PAC$ learning function classes~\cite{DBLP:conf/colt/BshoutyJ95,atici2007quantum,arunachalam2021two,arunachalam2018optimal} and showed some strengths and weakness of quantum examples in the $\PAC$ model of learning: under the distribution independent setting, we know that quantum examples are not useful for learning~\cite{DBLP:journals/jmlr/ArunachalamW18},   for uniform and product distributions we know quantum examples are useful~\cite{atici2007quantum,arunachalam2021two,DBLP:conf/colt/BshoutyJ95,kanade2018learning},\footnote{We remark that almost all known quantum speedups are based on a version of quantum Fourier  sampling.} for the uniform distribution we know they are not useful for learning circuit families~\cite{arunachalam2022quantum} and for certain applications such as learning parities with noise, quantum examples are known to be useful~\cite{Grilo2019learning}.  For further details, we refer the reader~to~\cite{DBLP:journals/sigact/ArunachalamW17}. 
\begin{question}
   Almost all quantum learning speedups are in the uniform distribution setting, is there a quantum learning speedup in the distribution-independent model in terms of sample or time~complexity?
\end{question}
%Its sample complexity is the maximum number of copies of $\ket{\psi_c}$ which the learner uses, over all concepts $c\in \Cc$, distributions $D$.   The $(\varepsilon,\delta)$-quantum sample complexity of a concept class $\Cc$ is the minimum sample complexity over all $(\varepsilon,\delta)$- quantum $\PAC$ learners $\A$ for $\Cc$.\snote{comment on this model is kind of like tomography/discrimination.} 

%

In~\cite{DBLP:journals/jmlr/ArunachalamW18} they also considered two other models of learning (motivated by classical computational learning theory): $(i)$ \emph{random classification noise learning}: here, the learner is given copies of $\sum_x \sqrt{D(x)} \ket{x} \otimes (\sqrt{1-\eta}\ket{c(x)}+\sqrt{\eta}\ket{\overline{ c}(x)}$ and the goal of the learning algorithm is the same as the $\PAC$ learner, $(ii)$ \emph{agnostic learning}: here $D:\01^{n+1}\rightarrow [0,1]$ is an unknown distribution,  the learner  is given copies of $\sum_{(x,b)\in \01^{n+1}} \sqrt{D(x,b)} \ket{x,b}$ and needs to find the concept $c\in \Cc$ that best approximates $D$, i.e., output $c$ that satisfies $\err_D(c)\leq \min_{c'\in \Cc} \{\err_D(c')\}+\varepsilon$, where $\err_D(c')=\Pr_{(x,b)\sim D}[c'(x)\neq b]$. In both these distribution-independent learning models,~\cite{DBLP:journals/jmlr/ArunachalamW18} showed that quantum sample complexity of learning is equal to classical sample complexity of learning up to  constant factors. A natural question is, what \emph{can} be learned in polynomial time in these models? As far as we are aware, only parities are known to be learnable in the classification noise model~\cite{Grilo2019learning,caro2020quantum} when $D=\01^n$ and agnostic learning interesting concept classes has not received any attention in~literature.
\begin{question}
      Can we learn DNF formulas in the quantum agnostic model in  polynomial time?\footnote{A positive answer to this question would imply a polynomial-time quantum algorithm  for $\PAC$ learning depth-$3$ circuits in the uniform distribution model~\cite{feldman2009distribution}.} 
      %Similarly, can parities be agnostic quantum learned in polynomial time?
\end{question}

\subsection{Statistical query model}
\label{sec:sqlearning}
The quantum statistical query model was introduced in~\cite{arunachalam2020quantum}, inspired by the classical statistical query model introduced by Kearns~\cite{kearns:statistical}. Classically, it is well-known that \emph{many} algorithms used in practice can be implemented using a statistical query oracle, for example, expectation maximization, simulated annealing, gradient descent, support vector machine, Markov chain Monte carlo methods, principal component analysis, convex optimization (see~\cite{reyzin2020statistical,feldman2017statistical} for these applications).  We first discuss the classical SQ model for learning an unknown concept $c$ from the concept class  $\Cc\subseteq \{c:\01^n\rightarrow \pmset{} \}$  under an unknown distribution $D:\01^n\rightarrow [0,1]$. The SQ learner  has access to a \emph{statistical query oracle} which takes as input two quantities:  \emph{tolerance} $\tau \geq 0$, a function $\phi:\01^n\times \pmset{}\rightarrow \pmset{}$ and returns $\alpha\in \R$ satisfying 
	$ 
	\Big |\alpha - \Exp_{x\sim D}[\phi(x,c(x))] \Big | \leq \tau\;.
	$
	The SQ learning algorithm adaptively chooses a sequence $\{ (\phi_i,\tau_i) \}$, and based on the responses of the statistical oracle $\{\alpha_i\}_i$,  it outputs an hypothesis $h:\01^n\rightarrow \pmset{}$ that approximates $c$, similar to the setting of $\PAC$ learning. 

The $\mathsf{QSQ}$ model is similar to the quantum $\PAC$ model, except that the learning algorithm isn't allowed entangled measurements on several copies of the quantum example state. More formally, let $\Cc\subseteq \{c:\01^n\rightarrow \pmset{}\}$ be a concept class,  $D:\01^n\rightarrow [0,1]$ be a distribution and let $\ket{\psi_c}=\sum_{x}\sqrt{D(x)}\ket{x,c(x)}$. In the $\mathsf{QSQ}$ model, a learning algorithm specifies an operator  $M$ satisfying $\|M\|\leq 1$,  tolerance $\tau \in [0,1]$ and obtains a number $\beta\in [\langle \psi_c| M |\psi_c\rangle-\tau ,\langle \psi_c| M |\psi_c\rangle+\tau]$. An intuitive way to think about the $\mathsf{QSQ}$ model is, a learning algorithm can specify a two-outcome measurement $\{M,\id-M\}$ and obtains a $\tau$-approximation of the probability of this measurement accepting $\ket{\psi_c}$. Ideally, one would want a $\mathsf{QSQ}$ algorithm for which $\tau=1/\poly(n)$ and $M$ can be implemented using $\poly(n)$ gates. A $\mathsf{QSQ}$ algorithm is amenable to near-term implementation since unlike the quantum $\PAC$ framework, it works only by making single copy measurements on the quantum example state $\ket{\psi_c}$. Surprisingly, in~\cite{arunachalam2020quantum}, they show that positive results for quantum $\PAC$ learning that we discussed in the previous section (such as learning parities, DNF formulas, $(\log n)$ juntas) can actually be implemented in the $\mathsf{QSQ}$ framework.
\begin{theorem}
\label{thm:everythingsq}
The concept classes consisting of parities, juntas, DNF formulas, sparse functions, can be learned under the uniform distribution in the $\mathsf{QSQ}$ model.
\end{theorem}
The crucial (and simple) observation in order to see this theorem is that computing the Fourier mass of a subset can be done in the $\mathsf{QSQ}$ model. Given that learning algorithms for parities, juntas, DNF formulas, sparse functions are via Fourier sampling~\cite{DBLP:journals/sigact/ArunachalamW17}, this observation implies the theorem. To see the observation, let $M=\sum_{S\in T}\ketbra{S}{S}$ and consider the observable
		$$
    M'=\mathsf{H}^{\otimes (n+1)} \cdot \Big( \id^{\otimes n}\otimes \ketbra{1}{1} \Big) \cdot M\cdot \Big( \id^{\otimes n}\otimes \ketbra{1}{1}\Big) \cdot \mathsf{H}^{\otimes (n+1)}.
		$$ 
    Operationally, $M'$ corresponds to first applying the
    Fourier transform on $\ket{\psi_f}$, {post-selecting on the last qubit being $1$} and finally applying $M$ to the first $n$ qubits.
 In order to see the action of $M'$ on~$\ket{\psi_f}$, first observe that $\mathsf{H}^{\otimes (n+1)}\ket{\psi_f}$ yields		
		$
		\frac{1}{\sqrt{2^{n}}}\sum_x\ket{x,f(x)}\rightarrow \frac{1}{2^n}\sum_{x,y}\sum_{b\in \01} (-1)^{x\cdot y+b\cdot f(x)}\ket{y,b}.
		$
        Conditioned on the $(n+1)$-th qubit being $1$, we have that the resulting quantum state is $\ket{\psi'_f}=\sum_Q \widehat{f}(Q)\ket{Q}$. The expectation value of $M$ with respect to the resulting state is given by
		$
		\langle\psi'_f\vert M\vert \psi'_f\rangle=\sum_{S\in T}\widehat{f}(S)^2.
		$
Therefore, one quantum statistical query with measurement $M'$,  tolerance $\tau$ produces a $\tau$-approximation of $\sum_{S\in T}\widehat{f}(S)^2$. In~\cite{arunachalam2020quantum}, they use this observation to prove Theorem~\ref{thm:everythingsq}. We pose the following question, which would serve as a tool to understand the fundamental question ``is entanglement needed for quantum learning Boolean functions?".\footnote{For learning the general class of quantum states, the recent work of Chen et al.~\cite{DBLP:conf/focs/ChenCH021} showed entanglement is needed for learning quantum states.}
\begin{question}
Is there a concept class separating $\mathsf{QSQ}$ and quantum  $\PAC$~learning with separable~measurements?
\end{question}
More recently, there have been few works that considered the ``diagonal-$\mathsf{QSQ}$" framework: here, the $\mathsf{QSQ}$ learner can only specify a \emph{diagonal} measurement operator $M$, i.e., the $\mathsf{QSQ}$ learner specifies a $\phi(x)\in [-1,1]$ and makes a $\mathsf{QSQ}$ query with $M=\sum_x \phi(x)\ketbra{x}{x}$ for the unknown state~$\ket{\phi}$. Recently~\cite{hinsche2021learnability,hinsche2022single,sweke23} looked at learning unknown circuits $U$ given diagonal-$\mathsf{QSQ}$ access to $\ket{\psi_U}=U\ket{0^n}$ (these learning algorithms allow to learn the output distributions $\{\langle x|U|0^n\rangle^2\}_x$ of unknown quantum circuits  $U$ in the computational basis).   In particular,~\cite{hinsche2022single} showed that distributions induced by Clifford circuits can be learned in the $\mathsf{QSQ}$ framework, however, if we add a single $T$ gate, then classical SQ learning the output distribution is as hard as learning parities with noise. Subsequent works~\cite{hinsche2021learnability,sweke23} looked at larger circuit families showing stronger lower bounds. One interesting question left open by their work is the following
\begin{question}
    What is $\mathsf{QSQ}$ complexity of learning output distributions of constant-depth circuits in the diagonal-$\mathsf{QSQ}$ framework? 
\end{question}
In another direction Du et al.~\cite{du2021learnability} showed that the $\mathsf{QSQ}$ model can be effectively simulated by noisy  quantum neural networks ($\textsf{QNN}$). Since we saw above that the $\mathsf{QSQ}$ model can learn certain concept classes  in polynomial time, their result suggests that $\textsf{QNN}$s implemented on a noisy device could potentially retain the quantum~speed-up.

\subsection{Kernel Methods}
So far we discussed a family of quantum algorithms that implicitly assumed the learning algorithm could learn a classical function by given access to quantum examples that encode classical information. Furthermore, these quantum examples use a number of qubits that is only logarithmic in the size of the unknown function. In this framework there have been several quantum  machine learning algorithms that are able to achieve polynomial or even exponential speed-ups over classical  approaches~\cite{Harrow2009quantum,wiebe2012quantum,lloyd2013quantum,Lloyd2014,Rebentrost2014quantum,lloyd2014quantumtda,Cong_2016,kerenidis2016quantum,Brandao2019sdp,Rebentrost2018svd,Zhao2019quantum}. However, it is not known whether data can be efficiently provided this way in practically relevant settings. This raises the question of whether the advantage comes from the quantum algorithm, or from the way data is provided~\cite{aaronson2015read}. Indeed, recent works have shown that if classical algorithms have an analogous sampling access to data, then some of the proposed exponential speed-ups do no longer exist~\cite{Tang2019quantuminspired,tang2018quantuminspired,gilyn2018quantuminspired,chia2018quantuminspired,ding2019quantuminspired,Chia2020samplingbased}. 

A natural question is, if we demand classical input and classical output, but let the intermediate operation be a quantum operations, can one hope for a quantum speedup? To this end, a powerful technique called the \emph{quantum kernel method} was introduced~\cite{Havlicek2019,Schuld2019quantum}. These papers proposed obtaining a quantum speedup via the use of a \emph{quantum-enhanced feature space}, where each data point is mapped non-linearly to a quantum state and then classified by a linear classifier in the high-dimensional Hilbert space.  The advantage of the quantum learner stems from its ability to recognize classically intractable complex patterns using the quantum feature map, which maps each classical data point non-linearly through a parameterized family of unitary circuits to a quantum state, $x\mapsto \ket{\phi(x)}=U(x)\ket{0^n}$, in both training and testing. The learning  algorithm proceeds by finding the optimal separating hyperplane for the training data in the high-dimensional feature space. To do so efficiently, they use the standard kernel method in \emph{support vector machines} (SVMs), a well-known family of supervised classification algorithms~\cite{vapnik2013nature}. More specifically, their algorithm only uses the quantum computer to estimate a kernel function and then implement a conventional SVM on a classical computer. In general, kernel functions are constructed from the inner products of the feature vectors for each pair of data points, which can be estimated as the transition amplitude of a quantum circuit as $\left|\braketIP{\phi(x_j)}{\phi(x_i)}\right|^2=|\langle 0^n|U^\dag (x_j)U(x_i)|0^n \rangle |^2$ (see Figure~\ref{fig:kernel} for the quantum circuit implementation of this). One can therefore estimate each kernel entry up to  a small additive error using the quantum computer -- a procedure that is referred to as \emph{quantum kernel estimation} ($QKE$). Then, the kernel matrix is given to a classical optimizer that efficiently finds the linear classifier that optimally separates the training data in feature space by running a convex quadratic program.  
%Kernel methods in machine learning are ubiquitous for pattern recognition, and we are optimistic that our techniques will find applications beyond classification problems.\snote{writebetter}

\begin{figure}[t]
    \centering
    \includegraphics[width=0.5\linewidth]{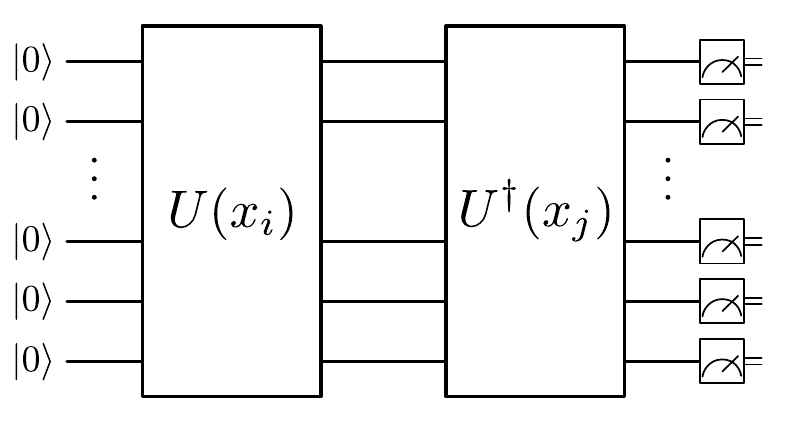}
    \caption{Quantum kernel estimation. A quantum feature map $x\mapsto \Phi(x) = \ketbra{\phi(x)}{\phi(x)}$ is represented by a circuit where $\ket{\phi(x)}=U(x)\ket{0^n}$. Each kernel entry $K(x_i,x_j)$ is obtained using a quantum computer by running the circuit $U^\dag (x_j)U(x_i)$ on input $\ket{0^n}$, and then estimating $\left|\langle 0^n|{U^\dag (x_j)U(x_i)}|{0^n}\rangle \right|^2$ by counting the frequency of the $0^n$ output.}
    \label{fig:kernel}
\end{figure}

Despite the popularity of these quantum kernel methods, it was unclear if, quantum algorithms using kernel methods could provide  a provable advantage over classical machine learning algorithms. There have been several proposal for interesting feature maps~\cite{glick2021covariant,Havlicek2019} based on group covariant maps, but their utility and implementability is unclear.  In~\cite{liu2021rigorous}, they constructed a classification task based on the discrete logarithm problem and showed that given classical access to data, quantum kernel methods can provably solve a classically intractable learning problem, even in the presence of finite sampling noise. While the particular problem they consider does require a fault-tolerant quantum computer, the learning algorithm is general in nature which suggests potential to find near-term implementable problems and is suitable for error-mitigation techniques. Their result can be viewed as one of the first formal evidence of quantum advantage using \emph{quantum kernel methods}, a widely-studied family of quantum learning algorithms that can be applied to a wide range of problems.

\begin{question}
Can quantum kernel methods give an unconditional polynomial advantage over classical learning for a natural problem?
\end{question}

\section{Perspective on other works}
\label{sec:prospective}
% In this survey we looked at rigorous theoretical results in learning quantum states that are centered around physical and computational questions. It is remarkable that we can obtain efficient algorithms - in some cases near optimal - for a wide class of learning problems, accompanied by new insights about some fundamental families of quantum states. 

% \paragraph{Other works.}
% Apart from the results cited above, there are several works related to learning quantum states that we haven't covered in this survey.  There has been works on learning quantum noise~\cite{flammia2012quantum}, quantum channels~\cite{huang2022learning,caro2022learning,chung2018sample,caro2021binary,fanizza2022learning,caro2020pseudo,haah2023query}, learning properties of ground states of the Hamiltonians~\cite{huang2022learning}, Born machines \cite{CMDK20, ZGYN22, GYN22} and learning unitaries and channels \cite{huang2022learning}. We also did not discuss learning unitaries defined by time evolution - $e^{iHt}$, where several recent works \cite{WiebeGFC14,haah2021optimal,HTFS22,dutt2021active} have given efficient algorithms. 
%The algorithm in~\cite{haah2021optimal} is similar to their high temperature Gibbs state algorithm, and utilizes cluster expansion and the work \cite{HTFS22} interleaves short time evolution with local unitary gates, allowing the algorithm to single out any desired Hamiltonian term and learn it. They achieve optimal sample complexity for the task. 

The theory of quantum information and computation lies at the intersection of computer science and physics. Quantum learning theory has evolved in the same spirit, addressing questions that are native to both bodies of knowledge. The field is inspired, on one hand, from the notions of $\PAC$ learning and statistical query learning from theoretical computer science and on the other hand, from the experimental goal of learning physics of a system from natural quantum states. This survey adopts the view that the most exciting questions in the field lie precisely at this intersection.

The success of learning theory lies in its adaptation of the `number of samples' as a natural complexity measure - which is well motivated from the point of view of practical machine learning. It is remarkable that we can obtain sample efficient algorithms - in many cases even time efficient - for a wide class of learning problems. These successful results are accompanied by new insights into the structure of corresponding families of quantum states, such as phase states, Gibbs quantum states and Matrix product states. Here, we list several notable works related to learning quantum states that haven't been covered in this survey. These include results on learning quantum noise~\cite{flammia2012quantum} in quantum experiments, tomography of quantum channels~\cite{huang2022learning,caro2022learning,chung2018sample,caro2021binary,fanizza2022learning,caro2020pseudo,haah2023query}, learning properties of ground states of the Hamiltonians~\cite{lewis2023improved,onorati2023efficient,huang2022provably,rouze2021learning}, provable bounds for learning parametrized quantum circuits, quantum dynamics, simulation~\cite{caro2021encoding,caro2022generalization,gibbs2022dynamical,caro2022out}, learning matrix product states~\cite{cramer2010efficient,GSG2023,khavari2021lower}, investigation into quantum Born machines~\cite{CMDK20, ZGYN22, GYN22}, learning Hamiltonians in a heuristic manner~\cite{wiebe2014hamiltonian,wiebe2012quantum,verdon2019quantum}, power of quantum neural networks~\cite{abbas2021power,beer2020training}, learning unitaries defined by time evolution - $e^{-iHt}$, where several recent works \cite{WiebeGFC14,haah2021optimal,HTFS22,dutt2021active} have given efficient algorithms. We refer the interested reader to the references for more details. 

\paragraph{Sample and time complexity beyond learnability.} 
Finally, we highlight that - beyond learnability - the notion of sample complexity is well motivated even in quantum information problems that are not canonical learning tasks. We discuss a few directions and questions here.

\emph{\textbf{1.}  Sample complexity as a measure in quantum communication.} In the standard quantum communication complexity setting, Alice and Bob compute a classical function with classical inputs, using quantum resources. One can also define a model where inputs are quantum and functions of quantum inputs are to be computed. An example of this is: Alice's input is a quantum state $\ket{\psi}$, Bob's input is a quantum state $\ket{\phi}$, and they wish to estimate $\bra{\psi} M \ket{\phi}$ for a given $M$. A single copy of each input is insufficient and unbounded number of inputs render the problem classical. An interesting intermediate regime is to allow several independent copies of inputs and minimize the sample complexity. The work \cite{ALL22} first considered this for $M=\id$ and showed exponential separation in sample complexity between classically communicating Alice-Bob and quantumly communicating Alice-Bob.
\begin{question}
What is the sample complexity of evaluating $\bra{\psi} M \ket{\phi}$ when Alice and Bob are only allowed classical communication, and how does it relate to the sample complexity when quantum communication is allowed? 
\end{question}

\emph{\textbf{2.} Sample complexity as a measure in the Local Hamiltonian problem:} A canonical Quantum Merlin-Arthur complete problem is the Local Hamiltonian Problem, with the goal of determining if the ground energy of a $n$-qubit local Hamiltonian is small or large. The proof - that certifies that the ground energy is small - is a quantum state, and there is some evidence that the proofs cannot be polynomial sized classical strings. In fact, the famous result of Marriot and Watrous shows that one copy of the witness suffices \cite{MW05}. Now, let's restrict the proof to be a simple quantum state, such as a state that can be prepared by a low-depth circuit or a stabilizer state. We can find local Hamiltonians whose ground states have very small overlap with one such state \cite{AN22}. Thus, many copies of the simple witness would be needed to eventually reach a complex witness of the ground state (via phase estimation algorithm). But it is not clear if such a simple witness could be useful in other ways to estimate the ground energy.
\begin{question}
Can we provide a sample complexity lower bound for interesting class of simple witness states, when the goal is to use them to estimate the ground energy of a Hamiltonian? Is this problem easier if the Hamiltonian itself is a sparse Hamiltonian with oracle access?
\end{question}

\emph{\textbf{3.} Time-efficient learning coset states.} One way to view the Hidden subgroup problem ($\mathsf{HSP}$) is in terms of sample complexity of learning the coset state. In the $\mathsf{HSP}$, there is a group $G$. Let $\Hi(G)$ be the set of all subgroups $H\leq G$ of $G$. We say a function $f_H:G\rightarrow S$ hides a subgroup $H$ if $f(x_1)=f(x_2)$ for all $x_1,x_2\in H$ and is distinct for different cosets. Given quantum query access to $f$, the goal is to learn $H$. The so-called \emph{standard approach} (which has been the focus of almost all known $\mathsf{HSP}$ algorithms) is the following: prepare $\frac{1}{\sqrt{|G|}}\sum_{x\in G}\ket{x}$, query $f$ to produce $\frac{1}{\sqrt{|G|}}\sum_{x\in G}\ket{x,f(x)}$ and discard the second register to obtain the state $\rho_H=\frac{|H|}{|G|}\sum_{g\in K}\ketbra{gH}{gH}$ where $\ket{gH}=\frac{1}{\sqrt{|H|}}\sum_{h\in H}\ket{gh}$ and $K$ is a complete set of left coset representatives of the subgroup $H\leq G$. The state $\rho_H$ is called the \emph{coset state} and the question is: what is the sample complexity and time complexity of learning $H$ given copies of $\rho_H$? A well-known result~\cite{DBLP:journals/ipl/EttingerHK04} shows that the sample complexity of learning $H$ is $O(\log^2 |G|)$. However, time-efficient learning $H$ for arbitrary groups has remained a long-standing open question. We know time efficient implements for special groups~\cite{DBLP:journals/cjtcs/BaconCD06,hallgren2000normal,DBLP:journals/eccc/ECCC-TR96-003,roetteler1998polynomial,Friedl2014hidden}. Given that learning $H$ given copies of $\rho_H$ for arbitrary groups has been open for decades, this motivates the following questions.

\begin{question}
For arbitrary groups, can we time-efficiently learn coset states in the alternate models of learning that we discussed in Section~\ref{sec:alternatemodels}? What other groups can we time-efficiently learn $H$ given copies of $\rho_H$?
\end{question} 

Additionally, we remark that all known $\mathsf{HSP}$ algorithms following the standard approach where they measure the second register, which leads to the following question.
\begin{question}
Does there exist a proposal for non-Abelian $\mathsf{HSP}$ that doesn't measure the second register in $\ket{\psi_f}=\frac{1}{\sqrt{|G|}}\sum_{x \in G}\ket{x,f(x)}$ and takes advantage of the function register to learn the unknown subgroup $H$? Similarly, can we extend the lower bounds in~\cite{DBLP:journals/jacm/HallgrenMRRS10} to the setting where the learning algorithm has access to copies of $\ket{\psi_f}$?
\end{question}

\emph{\textbf{4.} Sample complexity of generalizing LMR.}  Lloyd, Mohseni, and Rebentrost~\cite{Lloyd2014} understood the following question (in the context of Hamiltonian simulation): 
How many copies of an unknown quantum state $\rho$ are required to simulate a unitary $U=e^{-i\rho t}$ which encodes $\rho$ for some $t\in \R$? The LMR protocol~\cite{Lloyd2014} showed that the sample complexity of implementing $U$ up to diamond norm $\delta$ is $O(t/\delta^2)$, and has found several applications in quantum computing. Subsequently the sample complexity obtained by the LMR protocol was shown to be optimal~\cite{kimmel2017hamiltonian}. A natural followup question is the following.
\begin{question}
     What is the sample complexity of approximately implementing $e^{-i f(\rho) t}$ for other functions $f$ acting on density matrices?
\end{question}

\suppress{ 

Now, consider the goal of certifying that a local Hamiltonian has large ground energy, for example, greater than a number $E$. Intuitively, it should be even harder to certify this with quantum states, as one needs to make sure that no quantum state achieves low energy. In fact, we do not expect one copy of an $n$-qubit state to convince us that the Hamiltonian has energy at least $E$.  

Now, let's add a twist to this problem. Suppose we are given $k$ copies of an unknown quantum state $\ket{\psi}$ and we want an algorithm that accepts with probability very close to $1$, if $\ket{\psi}$ has energy $\geq E$. With $k=1$, this task can't be achieved, as various states have different energies. With very large $k$, we can simply perform tomography and then perform classical computation to compute the energy. Our question is as follows.
\begin{question}
What is the tight sample complexity for the problem of determining if the energy is at least $E$?
\end{question}
Such an algorithm could help certify that the ground energy is at least $E$ - if the success probability is always close to $1$, then the corresponding measurement operator has eigenvalues close to $1$, Note however, that it may not be a very efficient certificate.

However, it seems difficult to formally prove this. An interesting approach is instead to rule out various potential ways of certifying that the energy is at least $E$. Towards this, let's consider a quantum algorithm $\mathcal{A}$ that takes $k$ copies of a quantum state $\ket{\psi}$ and runs some energy measurement procedure on it (such as quantum phase estimation). If the energy is at least $E$, it accepts. Such a protocol can be viewed as performing a measurement $\{M, \id-M\}$ on $\ket{\psi}^{\otimes k}$ and the probability to accept is $\bra{\psi}^{\otimes k}M\ket{\psi}^{\otimes k}$. We allow the algorithm to accept with probability at least $c \geq 2^{-poly(n)}$. To determine that the probability of acceptance was close to 1 (guaranteeing that energy is indeed $\geq E$ for all states), we would have to compute the eigenvalues of the matrix $M'=\Pi^k_{\text{sym}}M\Pi^k_{\text{sym}}$. Here,  $\Pi^k_{\text{sym}}$ is the projector onto the symmetric subspace. This itself can be as hard a task as diagonalizing the Hamiltonian - except when $M'$ is exceedingly close to $\Pi_{\text{sym}}$ and any diagonal entry of $M'$ can be efficiently computed. Indeed, we can estimate $\frac{1}{d_k}\|M'-\Pi_{\text{sym}}\|^2_2 = \frac{1}{d_k}\sum_{i}(\bra{i}(M')^2\ket{i}-2\bra{i}M'\ket{i}+1)$, with $\{\ket{i}\}$ being a basis for the symmetric subspace and $d_k$ being its dimension. This leads to the next question.}

\textbf{Acknowledgements.} We thank Matthias Caro and  the anonymous reviews of Nature Reviews Physics for several comments improving the presentation of this work and Abhinav Deshpande for useful comments. We thank Iulia Georgescu for commissioning this survey for the Nature Reviews~Physics. AA acknowledges support through the NSF CAREER Award No. 2238836 and NSF award QCIS-FF: Quantum Computing \& Information Science Faculty Fellow at Harvard University (NSF 2013303).

%%%%%%%%%%%%%%%%%%%%%%%%%%%%%%%%%%%%%%%%%%%%%%%%%%%%%%%%%%%%%%%
\newcommand{\etalchar}[1]{$^{#1}$}

\end{document}